\documentclass[12pt,reqno]{amsart}

\usepackage{amsmath}
\usepackage{amsfonts}
\usepackage{amssymb}
\usepackage{amsbsy}
\usepackage{amsthm}
\usepackage{amscd}
\usepackage{color}
\usepackage{graphicx}
\usepackage[normalem]{ulem}

\usepackage{layout}

\newtheorem{thm}{Theorem}[section]

\newtheorem{lem}[thm]{Lemma}
\newtheorem{prop}[thm]{Proposition}
\theoremstyle{definition}

\theoremstyle{remark}

\newtheorem{ex}[thm]{Example}

\newcommand{\opn}[1]{\operatorname{#1}}

\newcommand{\mc}{\mathcal}

\def\bs{\boldsymbol}
\def\>{\rangle}
\def\<{\langle}
\def\ua{\uparrow}
\def\da{\downarrow}
\def\la{\leftarrow}
\def\ra{\rightarrow}
\def\re{\opn{Re}}
\def\im{\opn{Im}}
\def\tr{\opn{Tr}}
\def\rank{\opn{rank}}
\def\ker{\opn{ker}}
\def\Prob{\opn{Prob}}
\def\spn{\opn{span}}

\def\Xint#1{\mathchoice
  {\XXint\displaystyle\textstyle{#1}}%
  {\XXint\textstyle\scriptstyle{#1}}%
  {\XXint\scriptstyle\scriptscriptstyle{#1}}%
  {\XXint\scriptscriptstyle\scriptscriptstyle{#1}}%
  \!\int}
  \def\XXint#1#2#3{{\setbox0=\hbox{$#1{#2#3}{\int}$}
      \vcenter{\hbox{$#2#3$}}\kern-.47\wd0}} 
  \def\dashint{\Xint-}

\begin{document}


\title[QUANTUM SUBSPACE RECURRENCE AND SCHUR FUNCTIONS]{\bf Quantum recurrence of a subspace and operator-valued Schur functions}

\author{J. Bourgain, F.A. Gr\"{u}nbaum, L. Vel\'{a}zquez. J. Wilkening}

\address[J. Bourgain]{School of Mathematics \\ Institute for Advanced Study \\ Princeton \\ NJ \\ 08540}
\email{bourgain@math.ias.edu}

\address[F. A. Gr\"{u}nbaum, J. Wilkening]{Department of Mathematics \\ University of California \\ Berkeley \\ CA \\ 94720}
\email[F. A. Gr\"{u}nbaum]{grunbaum@math.berkeley.edu} \email[J. Wilkening]{wilken@math.berkeley.edu}

\address[L. Vel\'{a}zquez]{Departamento de Matem\'{a}tica Aplicada \\ Universidad de Zaragoza $\&$ IUMA \\ Zaragoza \\ Spain}
\email{velazque@unizar.es}

\subjclass[2000]{}


\date{}

\maketitle

\begin{abstract}

A notion of monitored recurrence for discrete-time quantum processes was recently introduced in \cite{GVWW} taking the initial state as an absorbing
one. We extend this notion of monitored recurrence to absorbing subspaces of arbitrary finite dimension.

The generating function approach leads to a connection with the well-known theory of operator-valued Schur functions. This is the cornerstone of a
spectral characterization of subspace recurrence that generalizes some of the main results in \cite{GVWW}. The spectral decomposition of the unitary
step operator driving the evolution yields a spectral measure, which we project onto the subspace to obtain a new spectral measure that is purely
singular iff the subspace is recurrent, and consists of a pure point spectrum with a finite number of masses precisely when all states in the
subspace have a finite expected return time.

This notion of subspace recurrence also links the concept of expected return time to an Aharonov-Anandan phase that, in contrast to the case of state
recurrence, can be non-integer. Even more surprising is the fact that averaging such geometrical phases over the absorbing subspace yields an integer
with a topological meaning, so that the averaged expected return time is always a rational number. Moreover, state recurrence can occasionally give
higher return probabilities than subspace recurrence, a fact that reveals once more the counterintuitive behavior of quantum systems.

All these phenomena are illustrated with explicit examples, including as a natural application the analysis of site recurrence for coined walks.

\end{abstract}

\section{Introduction} \label{sec:INT}

The study of return properties of discrete-time random processes goes back at least to G.~P\'{o}lya in 1921 \cite{POLYA}, who proved that only in
dimension one and two is the simplest unbiased random walk recurrent, returning to the starting point with probability one. Since then, the notion of
recurrence has been key in the theory of classical random processes.

Concerning the analog of recurrence for discrete-time quantum systems, two different approaches have been recently proposed. In 2008,
M.~\v{S}tefa\v{n}\'{a}k, I.~Jex and T.~Kiss \cite{SJK1,SJK2,SJK3} resort to an ensemble of identically prepared independent systems that are discarded
after the measurement to avoid changing the dynamics due to the collapse. A more recent proposal, from A.H.~Werner, R.F.~Werner and two of the
authors \cite{GVWW}, treats the perturbation of the evolution due to the monitoring as an essential ingredient in the definition of quantum
recurrence. This second definition of recurrence is based on performing a projective measurement after each step of the quantum evolution, an
idea that has some precursors in the context of quantum walks with absorbing boundary conditions \cite{ABNVW,BCGJW}.

The ideas developed in \cite{GVWW} deal with the probability of returning to the initial state, which corresponds to having an absorbing
one-dimensional subspace generated by this initial state. A natural generalization, which is the aim of this paper, is to consider an absorbing
subspace of arbitrary finite dimension which includes the initial state. This is the case, for instance, when we are interested in the return of a
coined walk to the initial site, regardless of the internal state. This generalization also allows us to ask about the probability of returning to a
finite set of sites from which an initial state was chosen. We will use the terminology ``state recurrence'' to refer to the notion of recurrence
defined in \cite{GVWW}, and ``subspace recurrence'' for the extension of this concept developed in the present paper.

The spectral characterization of subspace recurrence is the principal objective of this paper, which emerges from the identification of a key object
that best codifies subspace recurrence, namely the
operator-valued Schur
function of the spectral measure of the absorbing subspace. This connects the subject of quantum recurrence with well-known areas of complex
analysis: the theory of matrix Schur functions \cite{DFK,BD} and the theory of orthogonal polynomials on the unit circle \cite{SIMON1,SIMON2}. See
also \cite{SIMONmatrix} for a more detailed discussion of the matrix version of these orthogonal polynomials and its relation with matrix Schur
functions.

It turns out that the recurrent subspaces -- those whose states all return with probability one to the subspace -- are characterized by inner Schur
functions. This identifies the recurrent subspaces as those contained in the singular subspace of the unitary step operator driving the evolution.
Furthermore, the recurrent subspaces whose states all possess a finite expected return time are those whose Schur functions are rational inner. Such
subspaces are those contained in finite sums of eigenspaces of the unitary step. These results are the natural extensions of Theorems 1 and 2 in
\cite{GVWW}.

We stress that the aim of the paper is not limited to the generalization to subspace recurrence of the main results already known for state
recurrence. The paper also shows important differences and unexpected relations between subspace and state recurrence. For instance, in contrast to
state recurrence, the existence of states that return with probability one to a subspace of finite dimension greater than one does not require the
existence of a singular subspace for the unitary step. This is true even if the state returns to the subspace in a finite expected time.

Besides, in the case of state recurrence, the expected return time is always infinite or an integer. This is because the expected return time
has a topological interpretation as the degree of the related inner Schur function \cite{GVWW}. These results are no longer true for subspace
recurrence because the expected return time of a state is related to a geometrical (rather than a topological) invariant in this case,
namely the Aharonov-Anandan geometrical phase of the curve obtained by applying the (operator) boundary values of the Schur function
on the unit circle to the state. As a consequence, the expected return time to a subspace can be any real number not smaller than one.

However, a remarkable thing occurs when averaging the expected return time over the absorbing subspace: such an average is always infinite or a
rational number. Actually, the appearance of integer quantities in the finite case has again a topological explanation:
the numerator of the rational average is the degree of the determinant of the corresponding Schur function,
while the denominator is simply the dimension of the subspace. This integer numerator may also be expressed in terms of
invariant subspaces for the unitary step (as the dimension of the minimal invariant subspace containing the absorbing subspace), and in terms of the
spectral measure of the subspace (as the sum of the ranks of all the mass points).

We should also highlight a quantum paradox that arises from the comparison between state and subspace recurrence: the probability of returning to a
state can be bigger than to a subspace containing such a state. Indeed, a state can return to itself in a finite expected time without returning with
probability one to the subspace (see Example~\ref{ex:SHIFT+EIG-Schur}). Thus, in contrast
to the classical expectation, the return probability can decrease when the dimension of the absorbing subspace increases.
Examples of this for one- and two-dimensional coined quantum walks will be given in Section~\ref{sec:SITE} (see Figs.~\ref{fig:HL-L}--\ref{fig:HEXA}).

The use of Schur function methods to the study of quantum recurrence is not limited to the spectral characterization of recurrence. For instance, the
return probability to a subspace, as well as the expected return time, can be computed using the Schur function of the corresponding spectral measure
(see Section~\ref{sec:SITE}). We find that many techniques and results borrowed from the theory of orthogonal polynomials on the unit circle
are invaluable tools for computing this Schur function (see Appendix~\ref{app:SITE-SCHUR}).

Regarding the structure of the paper, Section~\ref{sec:SUB-REC} develops the notion of subspace recurrence, including its spectral characterization.
The expected time for the return of a state to a subspace is discussed in Section~\ref{sec:EXP-TIME}, which ends with a spectral characterization of
the subspaces whose states all return in a finite expected time. Section~\ref{sec:SITE} deals with applications of the previous results to the study
of site recurrence for coined quantum walks. Finally, Appendix~\ref{app:TAU-AVERAGE} includes the proofs of some technical results of interest for
Section~\ref{sec:EXP-TIME}, while Appendix~\ref{app:SITE-SCHUR} gives the details of the calculation of the matrix Schur function for a site in a
one-dimensional coined walk, which is used in the applications of Section~\ref{sec:SITE}.

\section{Subspace recurrence} \label{sec:SUB-REC}

We will consider a quantum discrete-time evolution driven by a unitary step operator $U$ on a complex separable Hilbert state space $\mathcal{H}$
with inner product $\langle\cdot|\cdot\rangle$ and norm $\|\cdot\|$. The rays of $\mc{H}$, which are represented up to phase factors by vectors on
the unit sphere $S_\mathcal{H}=\{\psi\in\mathcal{H} : \|\psi\|=1\}$, are the quantum states that evolve after $n$ steps via $\psi\to U^n\psi$.

Given a subspace $V$ of $\mathcal{H}$ and a state $\psi \in S_V=\{\psi\in V : \|\psi\|=1\}$, we can ask about the probability of returning to the
subspace $V$ if the evolution starts at the state $\psi$. Following the ideas developed in \cite{GVWW}, a monitored procedure to define such a
probability consists in introducing measurements to check the return of the system after every step.

More precisely, starting at $\psi=\psi_0$, any step $\psi_n \to U\psi_n$ is followed by a measurement of the orthogonal projection $P=P_V$ onto $V$.
When this projective measurement succeeds, the experiment ends.  Otherwise the system is projected to $V^\bot$ so that it is left in the
state $\psi_{n+1} \in S_\mathcal{H}$ obtained by normalizing $(I-P)U\psi_n$.

If the system was not found in the subspace $V$ during the first $n$ steps, after the $n$-th step it will be in a state proportional to
$\tilde{U}^n\psi$ with $\tilde{U}=(I-P)U$. According to the standard rules of Quantum Mechanics, the probability of this survival is
$s_n(\psi)=\|\tilde{U}^n\psi\|^2$, and we will call it the {\it $n$-step $V$-survival probability} of $\psi$.

It follows that
the probability that the system returns to the subspace $V$ for the first time in the $n$-th step comes from multiplying the
$n\!-\!1$-step survival probability $\|\tilde{U}^{n-1}\psi\|^2$ by the probability $\|PU\psi_{n-1}\|^2$ to find the system in the subspace $V$
conditioned on such a survival. Therefore, the probability we are searching for is given by $\|PU\tilde{U}^{n-1}\psi\|^2$, and we will refer to it as
the {\it $n$-step first $V$-return probability} of $\psi$.

We emphasize
the difference between the first $V$-return probabilities and the quantities $\|PU^n\psi\|^2$, which also
represent the probability of finding
the system in the subspace $V$, but performing the projective measurement only in the $n$-th step, i.e. without any intermediate monitoring
of the process. To distinguish both kind of probabilities we will refer to
$\|PU^n\psi\|^2$ as
the {\it $n$-step $V$-return probability} of $\psi$.

First $V$-return probabilities and $V$-return probabilities can be written as $\|\bs{a}_n\psi\|^2$ and $\|\bs\mu_n\psi\|^2$ respectively, where
$\bs{a}_n$ and $\bs\mu_n$ are the operators on $V$ given by
\begin{equation}\label{eq:an:def}
\begin{aligned}
 & \bs{a}_n = PU\tilde{U}^{n-1}P & & \quad \text{\it $n$-step first $V$-return amplitude operator},
 \\
 & \bs\mu_n = PU^nP & & \quad \text{\it $n$-step $V$-return amplitude operator}.
\end{aligned}
\end{equation}
We will use boldface characters to denote those operators that, like $\bs{a}_n$ and $\bs\mu_n$, should be considered as operators on $V$. Note that
$\bs{a}_n$ and $\bs\mu_n$, apart from $n$, depend (only) on the unitary step $U$ and the subspace $V$. If it is convenient to indicate explicitly
this dependence for any of the mathematical objects that arise, we will use a superscript and a subscript, writing for instance $\bs{a}_{n,V}^U$ and
$\bs\mu_{n,V}^U$.

The matrix elements of the above operators also provide transition amplitudes of interest. For any states $\psi,\phi \in S_V$, the squared modulus of
the amplitude $\<\phi|\bs{a}_n\psi\>$ gives the probability of the $n$-step transition $\psi\to\phi$ without hitting $V$ at intermediate steps, while
the squared modulus of the amplitude $\<\phi|\bs\mu_n\psi\>$ yields the corresponding $n$-step transition probability without intermediate
monitoring.

To avoid overcounting probabilities, the total probability of a transition -- the probability that it will take place eventually, without any
reference to when it takes place -- should be defined as the sum over the number of steps $n$ of the $n$-step first time transition probabilities.
This leads to the following total probabilities for any states $\psi,\phi\in S_V$,
\begin{equation} \label{eq:R:def}
\begin{aligned}
 & \kern-7pt R(\psi)=\sum_{n\ge1}\|\bs{a}_n\psi\|^2 & & \text{{\it \; $V$-return probability} of $\psi$},
 \\
 & \kern-9pt \Prob(\psi,\phi)=\sum_{n\ge1}|\<\phi|\bs{a}_n\psi\>|^2 & & \text{\it \; $\psi\to\phi$ probability when returning to $V$}.
\end{aligned}
\end{equation}
The second quantity measures the relative likelihood of landing on one state in $V$ compared to another. Note that $R(\psi)=\sum_k\Prob(\psi,\phi_k)$
for any orthonormal basis $\phi_k$ of $V$.

Quantum physicists, convinced by their expertise on the consistency of Born's probabilistic interpretation of Quantum Mechanics, would say that the
probabilistic nature of the above quantities ensures that $\sum_{n\ge1}\|\bs{a}_n\psi\|^2$ and $\sum_{n\ge1}|\<\phi|\bs{a}_n\psi\>|^2$ are not
greater than one. Nevertheless, this can be easily proved, showing that the probabilistic interpretation of the above series does not lead to a
contradiction.

The orthogonal decomposition $U\tilde{U}^{n-1}P = \bs{a}_n + \tilde{U}^nP$ yields
\begin{equation} \label{eq:an-sn}
 \|\bs{a}_n\psi\|^2 = \|\tilde{U}^{n-1}\psi\|^2 - \|\tilde{U}^n\psi\|^2, \qquad \psi\in V,
\end{equation}
thus the $V$-return probability of a state $\psi\in S_V$ is given by
\[
 R(\psi) = 1 - \lim_{n\to\infty}\|\tilde{U}^n\psi\|^2.
\]
This identity not only proves that $R(\psi)\leq1$, but also that $s_n(\psi)=\|\tilde{U}^n\psi\|^2$ converges to a number
$s(\psi)=\lim_{n\to\infty}\|\tilde{U}^n\psi\|^2 \leq 1$ that provides the total {\it $V$-survival probability} of $\psi$, i.e. the probability that
the state $\psi$ is never recaptured by the subspace $V$. The convergence of the $n$-step survival probabilities $s_n(\psi)$ to a number not greater
than one can also be proved directly by checking that $\|\tilde{U}^n\psi\|$ is non-increasing and bounded by one. Finally, the fact that
$|\<\phi|\bs{a}_n\psi\>| \le \|\bs{a}_n\psi\|$ for any $\phi \in S_V$ yields $\Prob(\psi,\phi) \le 1$.

%
%
%
%

\begin{ex} \label{ex:SHIFT+EIG}
Let $\mathcal{H}=\mathbb{C}^2\oplus\ell^2(\mathbb{Z})$ with an orthonormal basis whose states we denote by \mbox{$|\kern-3pt\ua\>$},
\mbox{$|\kern-3pt\da\>$} and $|x\>$, $x\in\mathbb{Z}$. We define a unitary operator $U$ on $\mathcal{H}$ as the shift $U|x\>=|x+1\>$ on the subspace
spanned by $|x\>$, $x\in\mathbb{Z}$, together with \mbox{$U|\kern-3pt\ua\>=|\kern-3pt\da\>$}, \mbox{$U|\kern-3pt\da\>=|\kern-3pt\ua\>$}.

Consider the subspace $V=\spn\{\psi,\phi\}$, where
\[
 \psi=|\kern-3pt\ua\>, \qquad \phi=\frac{1}{\sqrt{2}}(|\kern-3pt\da\>+|0\>).
\]
Then, \mbox{$\tilde{U}\psi=\frac{1}{2}(|\kern-3pt\da\>-|0\>)$} and \mbox{$\tilde{U}^n\psi=-\frac{1}{2}|n-1\>$} for $n\ge2$, while $\tilde{U}^n\phi=$
$\frac{1}{\sqrt{2}}|n\>$ for $n\ge1$. Therefore,
\[
 \begin{cases} \bs{a}_1\psi=\frac{1}{\sqrt{2}}\phi, \\ \bs{a}_1\phi=\frac{1}{\sqrt{2}}\psi, \end{cases}
 \quad
 \begin{cases} \bs{a}_2\psi=\frac{1}{2}\psi, \\ \bs{a}_2\phi=0, \end{cases}
 \quad \bs{a}_n=\bs0 \; \text{ for } \; n\ge3.
\]
For any $\psi_{\alpha,\beta}=\alpha\psi+\beta\phi\in S_V$ we find that $R(\psi_{\alpha,\beta})=\frac{1}{2}+\frac{1}{4}|\alpha|^2$, which reaches its
maximum 3/4 and minimum 1/2 at $\psi$ and $\phi$, respectively. The transition probabilities when returning to $V$ between these two extreme states
are $\Prob(\psi,\phi)=\Prob(\phi,\psi)=1/2$. \hfill $\scriptstyle\blacksquare$
\end{ex}

\subsection{Generating functions} \label{ssec:GEN-FUN}

The inequalities $\|\bs\mu_n\|,\|\bs{a}_n\|\le1$ ensure that the series
\[
 \hat{\bs\mu}(z)=\sum_{n\ge0}\bs\mu_nz^n, \qquad \hat{\bs{a}}(z)=\sum_{n\ge1}\bs{a}_nz^n, \qquad |z|<1,
\]
are absolutely convergent and hence define operator-valued functions on the unit disk. We will refer to $\hat{\bs\mu}(z)$ and $\hat{\bs{a}}(z)$ as
the $V$-return and first $V$-return generating functions, respectively.

These generating functions are related to the Green functions of the operators $U$ and $\tilde{U}$,
\[
 G(z)=(I-zU)^{-1}, \qquad \tilde{G}(z)=(I-z\tilde{U})^{-1},
\]
through the relations
\[
 \hat{\bs\mu}(z)=PG(z)P, \qquad \hat{\bs{a}}(z)=zPU\tilde{G}(z)P.
\]

The above relations, together with the resolvent identities
\[
 G(z)-\tilde{G}(z)=zG(z)PU\tilde{G}(z)=z\tilde{G}(z)PUG(z),
\]
and the equality $P\tilde{G}(z)P=P$, which follows from $\tilde{G}(z)-I=z\tilde{U}\tilde{G}(z)$, lead to
\[
 \hat{\bs\mu}(z)-P = \hat{\bs\mu}(z)\hat{\bs{a}}(z) = zPUG(z)P.
\]
Expanding $G(z)P=\tilde{G}(z)(P+zPUG(z)P)=\tilde{G}(z)P\hat{\bs\mu}(z)$ in the last term yields the renewal equations
\[
  \hat{\bs\mu}(z)\hat{\bs{a}}(z) = \hat{\bs{a}}(z)\hat{\bs\mu}(z) = \hat{\bs\mu}(z)-P,
\]
which are the extension to $\dim V>1$ of the renewal equations already found in \cite{GVWW} for $\dim V=1$. In terms of the Taylor
coefficients, these equations become
\[
 \bs{\mu}_n = \bs a_n + \sum_{k=1}^{n-1}\bs\mu_k\bs{a}_{n-k} = \bs a_n + \sum_{k=1}^{n-1}\bs{a}_{n-k}\bs\mu_k.
\]
In the classical case, an identical formula expresses the probability of returning in $n$ steps in terms of all the ways it can return for the first
time in $n$ or fewer steps. However, in the quantum case, $\bs a_n$ and $\bs\mu_n$ are (operator) \emph{amplitudes} rather than probabilities.
Nevertheless, both renewal equations hold despite the fact that the operator coefficients $\bs\mu_n$ and $\bs{a}_n$ do not commute in general. In any
case, they allow us to compute the first $V$-return amplitudes $\bs{a}_n$ from the $V$-return amplitudes $\bs\mu_n$, which are usually more
accesible.

Since we consider the coefficients $\bs\mu_n,\bs{a}_n$ as operators on $V$, the generating functions $\hat{\bs\mu}(z),\hat{\bs{a}}(z)$ become
functions with values in operators on $V$. Thus the term $P$ in the renewal equations must be understood as the identity $I_V$ in $V$, leading to
\[
 (I_V-\hat{\bs{a}}(z))\hat{\bs\mu}(z) = I_V.
\]
This shows that $\hat{\bs\mu}(z)$ and $I_V-\hat{\bs{a}}(z)$ are invertible on the unit disk with
\[
 \hat{\bs\mu}(z)=(I_V-\hat{\bs{a}}(z))^{-1}, \qquad \hat{\bs{a}}(z) = I_V-\hat{\bs\mu}(z)^{-1}.
\]
These relations permit us to express the $V$-return amplitudes purely in terms of the first $V$-return amplitudes and vice-versa,
\[
\begin{aligned}
 \bs\mu_n & = \bs{a}_n
 + \kern-3pt \sum_{\substack{n_1,n_2\ge1\\n_1+n_2=n}} \kern-7pt \bs{a}_{n_1} \bs{a}_{n_2}
 + \kern-7pt \sum_{\substack{n_1,n_2,n_3\ge1\\n_1+n_2+n_3=n}} \kern-15pt \bs{a}_{n_1} \bs{a}_{n_2} \bs{a}_{n_3}
 + \cdots + \bs{a}_1^n,
 \\
 \bs{a}_n & = \bs\mu_n
 - \kern-3pt \sum_{\substack{n_1,n_2\ge1\\n_1+n_2=n}} \kern-7pt \bs\mu_{n_1} \bs\mu_{n_2}
 + \kern-7pt \sum_{\substack{n_1,n_2,n_3\ge1\\n_1+n_2+n_3=n}} \kern-15pt \bs\mu_{n_1} \bs\mu_{n_2} \bs\mu_{n_3}
 - \cdots + (-1)^{n-1} \bs\mu_1^n.
\end{aligned}
\]

We can arrive at a spectral interpretation of the above generating functions from the spectral decomposition of the unitary operator $U$,
\[
 U = \int_{S^1} \lambda\,E(d\lambda), \qquad \int_{S^1} E(d\lambda) = I,
\]
which expresses $U$ as an integral over the unit circle $S^1=\{z\in\mathbb{C}:|z|=1\}$ with respect to the projector-valued spectral measure
$E(d\lambda)$. This allows us to rewrite its Green function as
\[
 G(z) = \int_{S^1} \frac{E(d\lambda)}{1-\lambda z}.
\]
The relation between $G(z)$ and $\hat{\bs\mu}(z)$ gives
\begin{equation} \label{eq:STIELTJES}
 \hat{\bs\mu}(z) = \int_{S^1} \frac{\bs\mu(d\lambda)}{1-\lambda z}, \qquad \bs\mu(d\lambda)=PE(d\lambda)P,
\end{equation}
where $\bs\mu(d\lambda)$, which must be considered as a measure with values in the space of
operators on $V$, will be called the {\it spectral measure of the
subspace $V$}. The measure $\bs\mu(d\lambda)$ is supported on a subset of the unit circle and, since $E(d\lambda)$ is a resolution of the identity,
$\int_{S^1}\bs\mu(d\lambda)=I_V$. For convenience, we will omit the domain of integration with respect to $\bs\mu(d\lambda)$ because it will
always be
$S^1$.

The equality \eqref{eq:STIELTJES} identifies $\hat{\bs\mu}(z)$ as the operator-valued Stieltjes function of the measure $\bs\mu(d\lambda)$. When $V$
is of finite dimension, identifying an operator on $V$ with a finite matrix allows us to use the results of the theory of matrix-valued Carath\'{e}odory
and Schur functions \cite{DFK,BD}. In what follows we will assume that $\dim V < \infty$.

With this assumption, the operator-valued Carath\'{e}odory and Schur functions of the measure $\bs\mu(d\lambda)$ are the analytic functions on the unit
disk defined by
\begin{equation}\label{eq:herglotz}
 \bs{F}(z) = \int \frac{\lambda+z}{\lambda-z}\,\bs\mu(d\lambda), \quad \bs{f}(z)=\frac{1}{z}(\bs{F}(z)-I_V)(\bs{F}(z)+I_V)^{-1}, \quad |z|<1.
\end{equation}
We will also refer to $\bs{F}(z)$ and $\bs{f}(z)$ as the Carath\'{e}odory and Schur functions of the subspace $V$. Evidently, $\bs{F}(0)=I_V$,
$\re\bs{F}(z) :=\frac{1}{2}(\bs{F}(z)+\bs{F}(z)^\dag)>0$ and $\bs{F}^\dag(z) = 2\hat{\bs\mu}(z)-I_V$, where $\bs{F}^\dag(z)=\bs{F}(\bar z)^\dag$ is
the analytic function on the unit disk whose Taylor coefficients are the adjoints of those of $\bs{F}(z)$.

The relation between Carath\'{e}odory and Schur functions gives
\begin{equation} \label{eq:reF-|f|}
\begin{aligned}
  & (I_V-z\bs{f}(z))(\bs{F}(z)+I_V) = 2I_V, \\
 & \re\bs{F}(z) = (I_V-\bar{z}\bs{f}(z)^\dag)^{-1} (I_V - |z|^2 \bs{f}(z)^\dag \bs{f}(z)) (I_V-z\bs{f}(z))^{-1},
 \\
 & \kern2pt I_V - |z|^2 \bs{f}(z)^\dag \bs{f}(z) = 4 (\bs{F}(z)^\dag+I_V)^{-1} \re\bs{F}(z) (\bs{F}(z)+I_V)^{-1} \ge 0,
\end{aligned}
\end{equation}
which shows that $\bs{f}(z)^\dag \bs{f}(z) \le I_V$ due to Schwarz's lemma.  This inequality is equivalent to $\bs{f}(z) \bs{f}(z)^\dag \le I_V$, and
can be also written as $\|\bs{f}(z)\| \le 1$.  If $\bs{f}(z)$ is not constant, the maximum principle implies that these inequalities are strict
inside the unit disk. Constant Schur functions appear for instance when $V$ is $U$-invariant (see below).

From (\ref{eq:reF-|f|}), we obtain $\hat{\bs\mu}(z) = (I_V-z\bs{f}^\dag(z))^{-1}$. Thus,
the renewal equation can be rewritten as
\[
 \hat{\bs{a}}(z) = z\bs{f}^\dag(z).
\]
This means that {\it the first $V$-return generating function is the Schur function of $V$, up to conjugation and multiplication by $z$}. In other
words, the first $V$-return amplitudes $\bs{a}_n$ are the adjoints of the Taylor coefficients for the Schur function $\bs{f}(z)$ of the subspace
$V$. In the case $\dim V=1$, this relation between the first $V$-return generating function and the Schur function of $V$ was reported for the
first time in \cite{GVWW}.

When $V$ is $U$-invariant, $P$ commutes with $U$ in (\ref{eq:an:def}) and $\hat{\bs{a}}(z)=zPUP$.
This identity reflects the obvious fact that any state of $V$ returns to $V$ in one step, and
$\bs{a}_1=U\upharpoonright V$ is the restriction of $U$ to $V$.  We also see that
$\bs{f}(z)=PU^\dag P$ is constant in this case.

The boundary behavior of $\bs{F}(z)$ as $|z|\to1$ provides relevant information about the measure: if
$\bs\mu(d\lambda)=\bs{w}(\theta)\frac{d\theta}{2\pi}+\bs\mu_{\text{sing}}(d\lambda)$, $\lambda=e^{i\theta}$, is its Lebesgue decomposition, then
$\bs{w}(\theta)=\lim_{r\to1^-}\re\bs{F}(re^{i\theta})$ for a.e.~$\theta$; the singular part $\bs\mu_{\text{sing}}(d\lambda)$ is concentrated on the
points $z\in S^1$ such that $\lim_{r\to1^-}\tr\re\bs{F}(rz)=\infty$; and the mass points $z\in S^1$ are characterized by a non-zero value of
$\bs\mu(\{z\})=\lim_{r\to1^-}\frac{1-r}{2}\bs{F}(rz)$. These results follow from the analogous ones for the scalar case using the polarization
identity for an operator $\bs{T}$ on a complex Hilbert space,
$\<\phi|\bs{T}\psi\>=\frac{1}{4}\sum_{k=0}^3\frac{1}{i^k}\<\phi+i\psi|\bs{T}(\phi+i\psi)\>$ ; see \cite{SIMONmatrix}.

Like Carath\'eodory functions, Schur functions can be defined a.e.~on the unit circle by their radial limits,
$\bs{f}(e^{i\theta})=\lim_{r\to1^-}\bs{f}(re^{i\theta})$. The following proposition gives a central result for the analysis of quantum
recurrence.

\begin{prop}
The measure $\bs\mu(d\lambda)$ is purely singular exactly when its Schur function $\bs{f}(z)$ is inner, i.e. when $\bs{f}(e^{i\theta})$ is
unitary for a.e.~$\theta$. In particular, $\bs\mu(d\lambda)$ is finitely supported -- it is a sum of finitely many mass points -- if and only if
$\bs{f}(z)$ is a rational inner function.
\end{prop}

\begin{proof}
Taking boundary values in \eqref{eq:reF-|f|} we find that $\re\bs{F}(e^{i\theta})=0$ is equivalent to
$\bs{f}(e^{i\theta})^\dag\bs{f}(e^{i\theta})=I_V$. Thus, up to a set of zero Lebesgue measure, $\bs{w}(\theta)=0$ if and only if
$\bs{f}(e^{i\theta})$ is unitary. When $\bs\mu(d\lambda)$ is finitely supported, $\bs{F}(z)$ and $\bs{f}(z)$ are obviously rational.
Conversely, a rational Schur function $\bs{f}(z)$ corresponds to a rational Carath\'{e}odory function $\bs{F}(z)$.
Since $\bs{F}(z)$ is rational, $\bs\mu_{\text{sing}}(d\lambda)$ must be
concentrated on the set of finitely many poles of $\bs{F}(z)$ that lie on the unit circle.  Since $\bs{f}(z)$ is inner,
$\bs{w}(\theta)=0$ a.e.~and $\bs\mu(d\lambda)=\bs\mu_\text{sing}(d\lambda)$ is finitely supported.
\end{proof}

\subsection{$V$-return probability and $V$-recurrence} \label{ssec:RET-REC}

The previous result is the key for the use of techniques of complex analysis in the study of quantum recurrence. A first consequence is that the
generating function $\hat{\bs{a}}(z)$ is also a Schur function with radial boundary values
$\hat{\bs{a}}(e^{i\theta})=e^{i\theta}\bs{f}(e^{-i\theta})^\dag$ a.e.~on the unit circle. The transition probabilities related to the subspace $V$ can
be expressed in terms of these boundary values as
\[
\begin{aligned}
 & \Prob(\psi,\phi) = \int_0^{2\pi} |\<\phi|\hat{\bs{a}}(e^{i\theta})\psi\>|^2  \, \frac{d\theta}{2\pi}
 = \int_0^{2\pi} |\<\phi|\bs{f}(e^{i\theta})^\dag\psi\>|^2 \, \frac{d\theta}{2\pi},
 \\
 & R(\psi) = \int_0^{2\pi} \|\hat{\bs{a}}(e^{i\theta})\psi\|^2 \, \frac{d\theta}{2\pi} = \<\psi|\bs{R}\psi\>,
 \quad \bs{R} \equiv \int_0^{2\pi} \kern-3pt \bs{f}(e^{i\theta}) \, \bs{f}(e^{i\theta})^\dag \, \frac{d\theta}{2\pi},
\end{aligned}
\]
where $R(\psi)$ and $\Prob(\psi,\phi)$ were defined in (\ref{eq:R:def}) above. Therefore, {\it the $V$-return probability
$R(\psi)=\<\psi|\bs{R}\psi\>$ is a quadratic form in $V$ given by the operator $\bs{R}$, which can be understood as the squared $L^2$ `right operator
norm' of the Schur function $\bs{f}(z)$}.

We will refer to $\bs{R}$ as the {\it $V$-return probability operator}, which obviously satisfies $0 \le \bs{R} \le I_V$ in agreement with the fact
that $\<\psi|\bs{R}\psi\>$ must be a probability for any $\psi\in S_V$. As a consequence, when $\psi$ runs over the unit sphere $S_V$, the possible
values of $R(\psi)$ cover the entire interval between the minimum and maximum eigenvalues of $\bs{R}$, and these eigenvalues
are the minimum and maximum $V$-return probabilities.
If these eigenvalues are distinct, i.e.\,$R(\psi)$ is not constant on $S_V$, the corresponding eigenspaces are orthogonal to each
other and the states of $V$ with maximum and minimum $V$-return probabilities are the unit spheres of these two eigenspaces.

The mean value of the $V$-return probability is the unbiased average of $R(\psi)$ over the states of $V$, namely
\[
 \dashint_{S_V} R(\psi) \, d\psi,
\]
with $d\psi$ the constant probability measure on $S_V$.
Averaging the diagonal coefficients $\<\psi|\bs{T}\psi\>$ of any normal operator $\bs{T}$ on $V$ over $S_V$ is
equivalent to averaging over any orthonormal basis of $V$.  Indeed, if $\phi_k$ is an orthonormal eigenvector basis of $\bs{T}$ and $\lambda_k$ are the
corresponding eigenvalues,
\begin{equation} \label{eq:AVERAGE}
 \dashint_{S_V} \<\psi|\bs{T}\psi\> \, d\psi = \sum_k \lambda_k \; \dashint_{S_V} |\<\phi_k|\psi\>|^2 \, d\psi
 = \frac{\sum_k\lambda_k}{\dim V} = \frac{\tr\bs{T}}{\dim V},
\end{equation}
where we have used that
\[
 \dashint_{S_V} |\<\phi|\psi\>|^2 \, d\psi
 = \frac{\|\phi\|^2}{\dim V}, \qquad \forall \phi\in\mathcal{H}.
\]
Applying this to the positive semidefinite operator $\bs{R}$ we find that
\[
 \dashint_{S_V} R(\psi) \, d\psi = \frac{\tr\bs{R}}{\dim V}
 = \frac{1}{\dim V} \int_0^{2\pi} \|\bs{f}(e^{i\theta})\|_F^2 \, \frac{d\theta}{2\pi},
\]
where $\|A\|_F=\tr(A^\dag A)$ is the Frobenius norm. This result states that {\it the mean $V$-return probability equals the arithmetic mean value of
the eigenvalues of the $V$-return probability operator $\bs{R}$, counting multiplicity.} Equivalently, the mean $V$-return probability comes from
averaging over the unit circle the arithmetic mean of the squared singular values of the Schur function $\bs{f}(z)$ of $V$, counting multiplicity.

We say that {\it a state $\psi\in S_V$ is $V$-recurrent if it returns to $V$ with probability one}, i.e. $R(\psi)=1$. In particular, the recurrence
of a state $\psi$, as defined in \cite{GVWW}, corresponds in our context to the special case that $V$ is the subspace spanned by $\psi$. One
of the main results in \cite{GVWW} is that the set of recurrent states is the unit sphere of the singular subspace of $U$. As we will see
below, when $\dim V>1$, the $V$-recurrent states continue to possess a subspace structure (up to normalization), but it need not be the singular
subspace of $U$.

Due to the identity $R(\psi) = \<\psi|\bs{R}\psi\>$ and the inequalities $0 \le \bs{R} \le I_V$, the existence of $V$-recurrent states is equivalent
to the condition $\|\bs{R}\|=1$. This means that 1 is an eigenvalue of $\bs{R}$ and the unit sphere of the corresponding eigenspace is the set of
$V$-recurrent states. Similar arguments show that a state $\psi$ is $V$-recurrent if and only if
$\|\hat{\bs{a}}(e^{i\theta})\psi\|=1$ for a.e.~$\theta$, which means that $\bs{f}(e^{i\theta}) \bs{f}(e^{i\theta})^\dag \psi = \psi$
except on a set of zero Lebesgue measure.  Therefore, the existence of
$V$-recurrent states implies that 1 is a singular value of $\bs{f}(e^{i\theta})$ a.e.  The converse is not true because the
corresponding eigenvectors of $\bs{f}(e^{i\theta}) \bs{f}(e^{i\theta})^\dag$ will generally depend on $\theta$.

If $\bs\mu(d\lambda)=\bs{w}(\theta)\frac{d\theta}{2\pi}+\bs\mu_\text{sing}(d\lambda)$, we know that the weight is given for a.e.
$\lambda=e^{i\theta}$ by $\bs{w}(\theta) = \re\bs{F}(\lambda)$. The existence of $V$-recurrent states implies that $\|\bs{f}(z)\|=1$ a.e.~in $S^1$,
which, in view of \eqref{eq:reF-|f|}, is equivalent to saying that the operator $\bs{w}(\theta)$ is singular for a.e.~$\theta$. Hence, the presence
of $V$-recurrent states requires an a.e.~singular weight for the absolutely continuous part of the spectral measure of $V$. The above comments
suggest that the converse is not true. (This is confirmed in Example~\ref{ex:SHIFT+EIG-Schur} below).

Thinking classically, we expect the $V$-return probability of a state to be non-decreasing when enlarging the subspace $V$. However,
in quantum recurrence, modifying the target subspace $V$ also changes the monitoring of the process. This has the effect of
dramatically altering
not only the possible return paths to $V$, but also the way in which each path contributes to the total $V$-return probability. The
reason for this is that Quantum Mechanics adds complex amplitudes or positive probabilities depending on whether the paths have the same target state
or orthogonal target states. This results in the absence of monotonicity with respect to the subspace $V$ for the $V$-return probability. In
particular, the recurrence of a state $\psi$ is in general not easily connected to its $V$-recurrence for a subspace $V$ that contains $\psi$.

We will see that the above effects already appear in simple models like the one introduced in Example~\ref{ex:SHIFT+EIG}.
We will
also use this example to illustrate the generating function approach to quantum recurrence. Although this approach could be avoided for such a toy
model, the analysis below will clarify in a simple setting some aspects that will reappear in the applications discussed in Section~\ref{sec:SITE},
where Schur functions become an indispensable tool for the analysis of recurrence.

\begin{ex} \label{ex:SHIFT+EIG-Schur}
Let us compute the generating functions for the subspace $V=\spn\{\psi,\phi\}$ of Example~\ref{ex:SHIFT+EIG}. Identifying any operator on
$V$ with its matrix representation in the basis $\{\psi,\phi\}$, and denoting by $I_2$ the $2\times2$ identity matrix, we have the $V$-return amplitudes
\[
 \bs\mu_0 = I_2,
 \quad
 \bs\mu_{2n-1} = \begin{pmatrix} 0 & \frac{1}{\sqrt{2}} \\ \frac{1}{\sqrt{2}} & 0 \end{pmatrix}
 \; \text{and} \;\;
 \bs\mu_{2n} = \begin{pmatrix} 1 & 0 \\ 0 & \frac{1}{2} \end{pmatrix}
 \; \text{for} \;\; n\ge1,
\]
which lead to the generating functions
\[
 \hat{\bs\mu}(z) = \frac{1}{1-z^2} \begin{pmatrix} 1 & \frac{z}{\sqrt{2}} \\ \frac{z}{\sqrt{2}} & 1-\frac{z^2}{2} \end{pmatrix},
 \quad
 \hat{\bs{a}}(z) = I_2-\hat{\bs\mu}(z)^{-1} = \begin{pmatrix} \frac{z^2}{2} & \frac{z}{\sqrt{2}} \\ \frac{z}{\sqrt{2}} & 0 \end{pmatrix}.
\]
This gives the $V$-return probability matrix
\[
 \bs{R} =
 \int_0^{2\pi}
 \begin{pmatrix} \frac{3}{4} & \frac{e^{-i\theta}}{2\sqrt{2}} \\ \frac{e^{i\theta}}{2\sqrt{2}} & \frac{1}{2} \end{pmatrix} \frac{d\theta}{2\pi} =
 \begin{pmatrix} \frac{3}{4} & 0 \\ 0 & \frac{1}{2} \end{pmatrix},
\]
in agreement with the $V$-return probability obtained in Example~\ref{ex:SHIFT+EIG},
\[
 R(\psi_{\alpha,\beta}) =
 \begin{pmatrix} \overline\alpha & \overline\beta  \end{pmatrix} \bs{R} \begin{pmatrix} \alpha \\ \beta \end{pmatrix} =
 \frac{1}{2} + \frac{1}{4}|\alpha|^2,
 \qquad
 \psi_{\alpha,\beta}=\alpha\psi+\beta\phi.
\]
The corresponding Carath\'{e}odory and Schur functions
\[
 \bs{F}(z) = \frac{1}{1-z^2} \begin{pmatrix} 1+z^2 & \sqrt{2}z \\ \sqrt{2}z & 1 \end{pmatrix},
 \qquad
 \bs{f}(z) = \begin{pmatrix} \frac{z}{2} & \frac{1}{\sqrt{2}} \\ \frac{1}{\sqrt{2}} & 0 \end{pmatrix},
\]
yield the spectral measure of $V$
\[
 \bs\mu(d\lambda) =
 \begin{pmatrix} 0 & 0 \\ 0 & \frac{1}{2} \end{pmatrix} \frac{d\theta}{2\pi} +
 \frac{1}{2} \begin{pmatrix} 1 & \frac{1}{\sqrt{2}} \\ \frac{1}{\sqrt{2}} & \frac{1}{2} \end{pmatrix} \delta(\lambda-1) +
 \frac{1}{2} \begin{pmatrix} 1 & \frac{-1}{\sqrt{2}} \\ \frac{-1}{\sqrt{2}} & \frac{1}{2} \end{pmatrix} \delta(\lambda+1),
\]
where $\lambda=e^{i\theta}$. Note that the weight $\bs{w}(\theta)$ is singular despite the absence of $V$-recurrent states.
This occurs because 1 is a singular value of
$\bs{f}(e^{i\theta})$ but the corresponding eigenvector of $\bs{f}(e^{i\theta})\bs{f}(e^{i\theta})^\dag$ is $\theta$-dependent.

\begin{figure}
 \includegraphics[width=8cm]{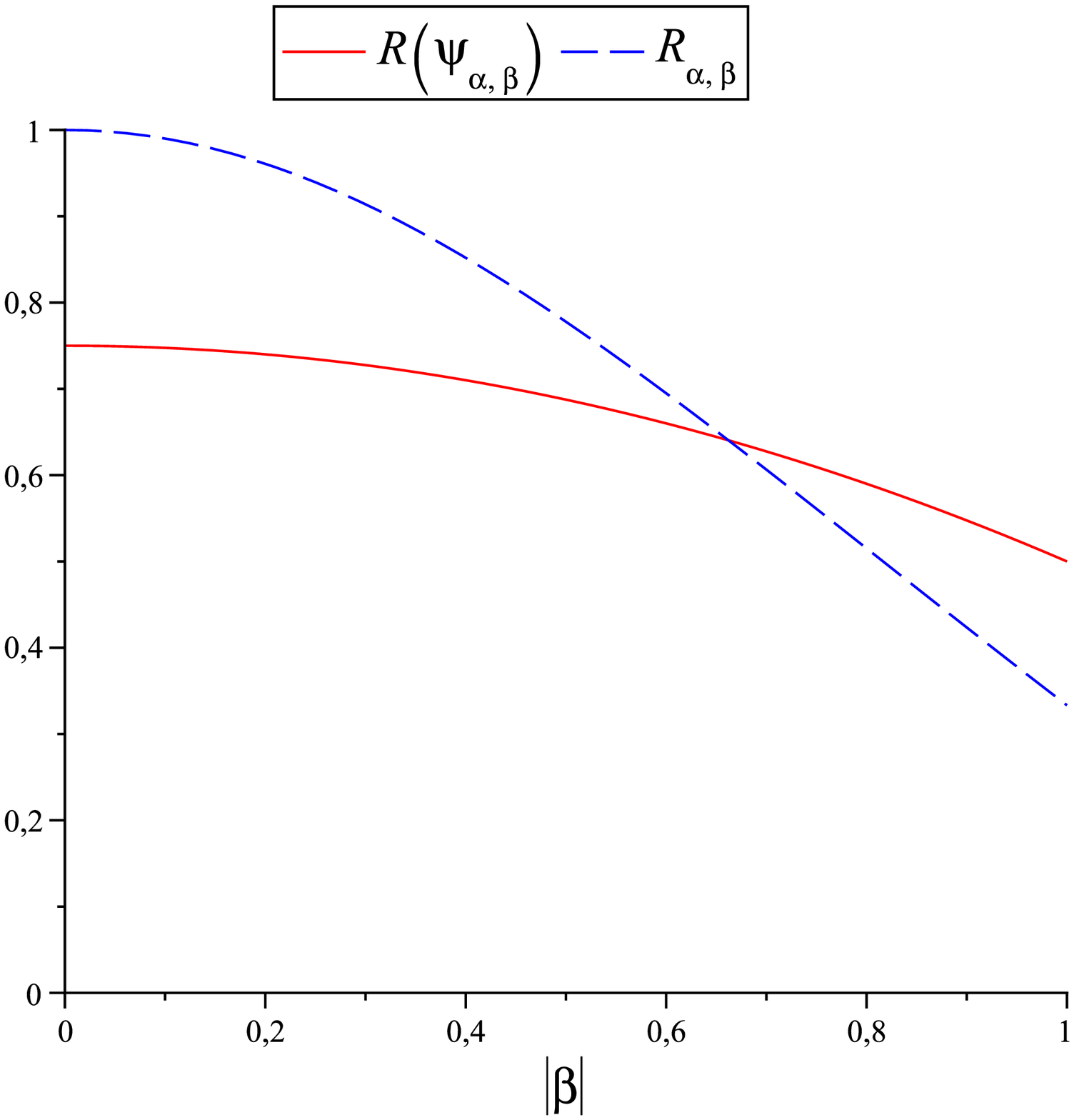}
 \caption{{\bf State versus subspace recurrence in Example~\ref{ex:SHIFT+EIG-Schur}}.
 The return probability for $\psi_{\alpha,\beta}=\alpha\psi+\beta\phi$ as a function of $|\beta|$.
 The return probability $R_{\alpha,\beta}$ to the state (dashed blue line)
 is higher than the return probability $R(\psi_{\alpha,\beta})$ to the subspace $V=\spn\{\psi,\phi\}$ (solid red line)
 when $|\beta|<\sqrt{\frac{1}{2}(5-\sqrt{17})}\approx0.66$.}
 \label{fig:state-sub-REC}
\end{figure}

We can compare state and subspace recurrence by computing the return probability $R_{\alpha,\beta}$ of the state $\psi_{\alpha,\beta}$. The scalar generating functions
\[
\begin{aligned}
 & \hat\mu_{\alpha,\beta}(z) =
 \begin{pmatrix} \overline\alpha & \overline\beta  \end{pmatrix} \hat{\bs\mu}(z) \begin{pmatrix} \alpha \\ \beta \end{pmatrix} =
 \frac{1-\frac{1}{2}|\beta|^2z^2+\sqrt{2}\re(\overline\alpha\beta)z}{1-z^2},
 \\
 & \hat{a}_{\alpha,\beta}(z) = 1-\frac{1}{\hat\mu_{\alpha,\beta}(z)} =
 \frac{\left(1-\frac{1}{2}|\beta|^2\right)z^2+\sqrt{2}\re(\overline\alpha\beta)z}{1-\frac{1}{2}|\beta|^2z^2+\sqrt{2}\re(\overline\alpha\beta)z},
\end{aligned}
\]
allow us to calculate
\[
 R_{\alpha,\beta} = \int_0^{2\pi} |\hat{a}_{\alpha,\beta}(e^{i\theta})|^2 \, \frac{d\theta}{2\pi} =
 \frac{1-\frac{1}{2}|\beta|^2}{1+\frac{1}{2}|\beta|^2},
\]
which should be compared to $R(\psi_{\alpha,\beta})=\frac{3}{4}-\frac{1}{4}|\beta|^2$ (see Figure~\ref{fig:state-sub-REC}). We find that
$R_{\alpha,\beta} > R(\psi_{\alpha,\beta})$ when $|\beta|^2 < \frac{1}{2}(5-\sqrt{17}) \approx 0.44$, which shows that the return probability to a state
can be bigger than the return probability to a subspace containing such a state. \emph{In fact, $\psi \in V$ is a recurrent state while $V$ has no
$V$-recurrent states!}  Furthermore, $\psi$ is not $V$-recurrent despite the fact that the expected time for the return of $\psi$ to itself is finite
and equal to 2 (see Section~\ref{sec:EXP-TIME}). \hfill $\scriptstyle\blacksquare$
\end{ex}

The previous example shows explicitly that a subspace $V$ with recurrent states need not have any $V$-recurrent states. On the other hand, it is easy to
check that the subspace $V=\spn\{|0\>,|1\>\}$ in Example~\ref{ex:SHIFT+EIG} has no recurrent states while $|0\>$ is $V$-recurrent. Thus, a subspace
with $V$-recurrent states need not have any recurrent states. Therefore, state recurrence and subspace recurrence seem to be unrelated notions.
Nevertheless, we will see that, surprisingly, demanding all the states of $V$ to be recurrent or $V$-recurrent are equivalent
requirements.

\begin{thm} \label{thm:REC}
Let $U$ be a unitary step and $V$ a finite-dimensional subspace with spectral measure $\bs\mu(d\lambda)$ and Schur function $\bs{f}(z)$. Then, the
following statements are equivalent:
\begin{enumerate}
\item All the states of $V$ are $V$-recurrent.
\item All the states of $V$ are recurrent.
\item $\bs\mu(d\lambda)$ is purely singular.
\item $\bs{f}(z)$ is inner.
\end{enumerate}
\end{thm}

\begin{proof}
Bearing in mind that $\bs{f}(e^{i\theta}) \bs{f}(e^{i\theta})^\dag \le I_V$ for a.e.~$\theta$ we find that
\[
\begin{aligned}
 R(\psi)=1, \, \forall\psi\in V & \,\Leftrightarrow\, \bs{R}=I_V \,\Leftrightarrow\, \bs{f}(z) \text{ is inner} \,\Leftrightarrow\,
 \\
 & \,\Leftrightarrow\, \bs\mu(d\lambda) \text{ is purely singular}.
\end{aligned}
\]
On the other hand, the spectral measure of a state $\psi$ is given by $\mu_\psi(d\lambda)=\<\psi|E(d\lambda)\psi\>$, where $E(d\lambda)$ is the
spectral measure of $U$. If $\psi\in S_V$, then clearly $\mu_\psi(d\lambda)=\<\psi|\bs\mu(d\lambda)\psi\>$. Hence, the polarization identity
shows that $\bs\mu(d\lambda)$ is purely singular if and only if $\mu_\psi(d\lambda)$ is purely singular for any $\psi\in S_V$. This
is equivalent to stating that all the states of $V$ are recurrent.
\end{proof}

We will say that {\it a subspace $V$ is recurrent if every state $\psi\in S_V$ returns to $V$ with probability one}, i.e.~if
every state $\psi\in S_V$ is
$V$-recurrent. This notion obviously coincides with the mere existence of $V$-recurrent states for the case $\dim V=1$ defining state recurrence;
otherwise, the presence of $V$-recurrent states does not necessarily mean that they exhaust all of $S_V$.

The above theorem is the generalization of \cite[Theorem 1]{GVWW} to subspace recurrence. It states that the finite-dimensional recurrent subspaces
are exactly those whose states are all recurrent. As a consequence, any subspace of a finite-dimensional recurrent subspace is recurrent too. Another
way to restate the previous result is to say that the finite-dimensional recurrent subspaces are those contained in the singular subspace of the
unitary step $U$.

We stress that despite the close analogy between Theorem \ref{thm:REC} and \cite[Theorem 1]{GVWW}, new situations appear when extending the concept
of state recurrence to subspaces, even if one restricts oneself to a finite-dimensional subspace $V$. In contrast to the case of state recurrence,
the presence of $V$-recurrent states does not imply the existence of a singular subspace for $U$ since such $V$-recurrent states need not cover all
of $S_V$. This is clearly shown, for instance, by the shift $U|x\>=|x+1\>$, $x\in\mathbb{Z}$, in $\ell^2(\mathbb{Z})$, since the state $|0\>$ is
$V$-recurrent for $V=\spn\{|0\>,|1\>\}$, but $U$ has no singular subspace. Moreover, the existence of a singular part for the measure
$\bs\mu(d\lambda)$ of a subspace $V$ does not ensure the presence of $V$-recurrent states. Example~\ref{ex:SHIFT+EIG-Schur} illustrates this fact
because $V=\spn\{\psi,\phi\}$ is free of $V$-recurrent states, while the related measure $\bs\mu(d\lambda)$ has two mass points, but is not purely
singular due to the absolutely continuous part, $|\phi\>\<\phi|\frac{d\theta}{4\pi}$. Additional examples of these phenomena will appear in
Section~\ref{sec:SITE}.

\section{Expected return time to a subspace} \label{sec:EXP-TIME}

The expected $V$-return time of a state $\psi\in S_V$, defined by means of
\[
 \tau(\psi) = \sum_{n\ge1}n\|\bs{a}_n\psi\|^2,
\]
makes sense whenever $\psi$ is $V$-recurrent, i.e.~$\sum_{n\ge1}\|\bs a_n\psi\|^2=1$. Otherwise, $\tau(\psi)$ is naturally
taken to be infinite.

Like the $V$-return probability $R(\psi)$, the expected $V$-return time of any state $\psi\in V$ can be expressed in terms of the $V$-survival
probabilities $s_n(\psi)=\|\tilde{U}^n\psi\|^2$. Using once again the relation \eqref{eq:an-sn} we find that
\[
  \sum_{n=1}^N n\|\bs a_n\psi\|^2 = \sum_{n=0}^{N-1}\|\tilde U^n\psi\|^2 - N\|\tilde U^N\psi\|^2.
\]
Therefore, the series on the right hand side must be divergent whenever $\tau(\psi)=\infty$. On the other hand, Lemma \ref{lem:SURVIVAL-2} of
Appendix \ref{app:SITE-SCHUR} shows that $\tau(\psi)<\infty$ implies $\lim_{N\to\infty}N\|\tilde{U}^N\psi\|^2=0$. Hence, the equality
\[
 \tau(\psi) = \sum_{n\geq0} \|\tilde{U}^n\psi\|^2
\]
always holds, which means that {\it the expected $V$-return time of a $V$-recurrent state is equal to the sum of its $n$-step $V$-survival
probabilities}.

We can also express $\tau(\psi)$ in terms of the Schur function $\bs{f}(z)$ of $V$. For any $V$-recurrent state $\psi$ we can write
\[
\begin{aligned}
 \tau(\psi)
 & = \lim_{r\to1^-} \int_0^{2\pi} \<\psi|\hat{\bs{a}}(re^{i\theta})^\dag\partial_\theta\hat{\bs{a}}(re^{i\theta})\psi\> \, \frac{d\theta}{2\pi i}
 \\
 & = 1 + \delta\tau(\psi),
 \quad
 \delta\tau(\psi) =
 \lim_{r\to1^-} \int_0^{2\pi} \<\psi|\partial_\theta\bs{f}(re^{i\theta})\bs{f}(re^{i\theta})^\dag\psi\> \, \frac{d\theta}{2\pi i}.
\end{aligned}
\]
In particular, if $V$ is recurrent and its inner Schur function $\bs{f}(z)$ has an analytic extension to a neighborhood of the closed unit disk,
e.g.~if $\bs{f}(z)$ is a rational inner function, then we can write
\begin{equation} \label{eq:BERRY}
 \tau(\psi) = \int_0^{2\pi} \<\psi(\theta)|\partial_\theta\psi(\theta)\> \, \frac{d\theta}{2\pi i},
 \qquad \psi(\theta)=\hat{\bs{a}}(e^{i\theta})\psi,
\end{equation}
where $\psi(\theta)$, $\theta\in[0,2\pi]$, traces out a closed curve on the sphere $S_V$ due to the unitarity of $\hat{\bs{a}}(e^{i\theta})$. This
simple result has a nice interpretation since it relates $\tau(\psi)$ to a kind of {\bf Berry's geometrical phase} \cite{BERRY,SIMON-Berry}. More
precisely, {\it the expected $V$-return time of a state $\psi\in S_V$ is $-1/2\pi$ times the {\bf Aharonov-Anandan phase} \cite{AA} associated
with the loop $\hat{\bs{a}}(e^{i\theta})\psi \colon S^1 \to S_V$}.

In the case of state recurrence, \cite{GVWW} proves that the states $\psi$ with a finite expected return time are characterized by a finitely
supported spectral measure $\mu_\psi(d\lambda)$, thus by a rational inner Schur function $f_\psi(z)$. Further, \cite{GVWW} also finds that $\tau(\psi)$
must be a positive integer whenever it is finite because of its topological meaning: $\tau(\psi)$ is the winding number of $\hat{a}_\psi(e^{i\theta})
\colon S^1 \to S^1$, where $\hat{a}_\psi(z)=z\overline{f_\psi}(z)$ is the first return generating function of $\psi$.

In contrast to a winding number, the Aharonov-Anandan phase is not necessarily an integer because it reflects a geometric rather than a topological
property of a closed curve. The expression \eqref{eq:BERRY} for $\tau(\psi)$ is reparametrization invariant, and changes by an integer under closed
$S^1$ gauge transformations $\psi(\theta)\to\tilde\psi(\theta)=e^{i\varsigma(\theta)}\psi(\theta)$, $\tilde\psi(2\pi)=\tilde\psi(0)$. This means that
$\tau(\psi)$ is a geometric property of the unparametrized image of $\psi(\theta)$ in $S_V$, while $e^{i2\pi\tau(\psi)}$ is a geometric property of
the corresponding closed curve in the projective space of rays of $S_V$ whose elements are the true physical states of $V$. In fancier language,
$S_V$ is a fiber bundle over such a projective space with structure group $S^1$, and $e^{-i2\pi\tau(\psi)}$ is the holonomy transformation associated
with the usual connection given by the parallel transport defined by $\<\psi(t)|\partial_t\psi(t)\>=0$ \cite{AA}.

As a consequence, we cannot expect for $\tau(\psi)$ to be an integer for subspaces $V$ of dimension greater than one. The following example shows
explicitly that the expected $V$-return time can be non-integer valued, a fact that will also be illustrated by the applications discussed
in Section~\ref{sec:SITE}. Like Example~\ref{ex:SHIFT+EIG-Schur}, the following example demonstrates the generating
function approach to computing the expected $V$-return time on a simple problem where other methods are available.

\begin{ex} \label{ex:FINITE}
Given a three-dimensional Hilbert space $\mathcal{H}$ with orthonormal basis $\{|0\>,|1\>,|2\>\}$, let $U$ be the cyclic shift $U|0\>=|1\>$,
$U|1\>=|2\>$, $U|2\>=|0\>$. Any subspace of $\mathcal{H}$ is recurrent in this case because the corresponding spectral measure is necessarily
finitely supported. Therefore, the expected $V$-return time makes sense for any state $\psi_{\alpha,\beta}=\alpha|0\>+\beta|1\>\in S_V$ of the
subspace $V=\spn\{|0\>,|1\>\}$. Then, $\tilde{U}\psi_{\alpha,\beta}=\beta|2\>$ and $\tilde{U}^n\psi_{\alpha,\beta}=0$ for $n\ge2$, while the first
$V$-return amplitudes are given by
\[
 \bs{a}_1\psi_{\alpha,\beta} = \alpha|1\>, \qquad \bs{a}_2\psi_{\alpha,\beta} = \beta|0\>, \qquad \bs{a}_n = \bs0 \; \text{ for } \; n\ge3.
\]
We conclude that
\[
 \tau(\psi_{\alpha,\beta}) = 1+|\beta|^2,
\]
which is an integer only for the states $|0\>$ and $|1\>$.

We can recover this result using the generating function approach. Identifying any operator on $V$ with its matrix representation in the basis
$\{|0\>,|1\>\}$, the $V$-return amplitudes
\[
 \bs\mu_{3n} = I_2,
 \quad
 \bs\mu_{3n+1} = \begin{pmatrix} 0 & 0 \\ 1 & 0 \end{pmatrix},
 \quad
 \bs\mu_{3n+2} = \begin{pmatrix} 0 & 1 \\ 0 & 0 \end{pmatrix},
 \quad
 n\ge0.
\]
yield the generating functions
\[
 \hat{\bs\mu}(z) = \frac{1}{1-z^3} \begin{pmatrix} 1 & z^2 \\ z & 1 \end{pmatrix},
 \qquad
 \hat{\bs{a}}(z) = \begin{pmatrix} 0 & z^2 \\ z & 0 \end{pmatrix}.
\]
Therefore, we can recover the previous expected $V$-return time from
\[
 \tau(\psi_{\alpha,\beta}) =
 \begin{pmatrix} \overline\alpha & \overline\beta \end{pmatrix}
 \bs\tau
 \begin{pmatrix} \alpha \\ \beta \end{pmatrix},
 \quad
 \bs\tau =
 \int_0^{2\pi} \hat{\bs{a}}(e^{i\theta})^\dag \partial_\theta\hat{\bs{a}}(e^{i\theta}) \frac{d\theta}{2\pi i} =
 \begin{pmatrix} 1 & 0 \\ 0 & 2 \end{pmatrix}.
\]

We can also compute the average of the expected $V$-return time over the sphere of states $S_V$. If $\sigma(d\alpha\,d\beta)$ is the constant
probability measure on the 3-sphere $S^3=\{(\alpha,\beta)\in\mathbb{C}^2 : |\alpha|^2+|\beta|^2=1\}$,
\[
 \dashint_{S_V} \tau(\psi)\,d\psi = \int_{S^3} (1+|\beta|^2)\,\sigma(d\alpha\,d\beta) = \frac{3}{2},
\]
which is not an integer but a rational number. \hfill $\scriptstyle\blacksquare$
\end{ex}

The above example suggests that, although the geometrical phases giving the expected $V$-return time can be arbitrary, their average over the states
of $V$ is always a rational number. We will see that this is indeed the case, but before doing that we will express such an average in a more
convenient way.

Using the monotone convergence theorem for $\dashint_{S_V} d\psi$ and the relation \eqref{eq:AVERAGE} we find that the averaged expected $V$-return
time is
\[
\begin{aligned}
 \dashint_{S_V} \tau(\psi) \, d\psi
 & = \sum_{n\ge1} \; \dashint_{S_V} n\|\bs{a}_n\psi\|^2 \, d\psi = \frac{1}{\dim V} \sum_{n\ge1}n\|\bs{a}_n\|^2_F
 \\
 & = \frac{1}{\dim V} \lim_{r\to1^-} \int_0^{2\pi} \tr\left(\hat{\bs{a}}(re^{i\theta})^\dag \partial_\theta \hat{\bs{a}}(re^{i\theta})\right)
 \, \frac{d\theta}{2\pi i},
 \\
 & = 1 + \frac{1}{\dim V} \lim_{r\to1^-} \int_0^{2\pi} \tr\left(\partial_\theta \bs{f}(re^{i\theta}) \bs{f}(re^{i\theta})^\dag\right)
 \, \frac{d\theta}{2\pi i},
\end{aligned}
\]
where we assume that the subspace $V$ is recurrent.

Apart from its own interest, this average is a useful tool to characterize those recurrent subspaces $V$ with $\tau(\psi)<\infty$ for all $\psi\in
S_V$: they are exactly the subspaces $V$ with finite average $\dashint_{S_V}\tau(\psi)\,d\psi$ due to the relation $\|\bs{a}_n\psi\|^2 \le
\|\bs{a}_n\|_F^2 = \sum_k\|\bs{a}_n\phi_k\|^2$ for any state $\psi$ and orthonormal basis $\phi_k$ of $V$. This means that the presence of a single
state $\psi\in S_V$ with $\tau(\psi)=\infty$ implies the divergence of the average $\dashint_{S_V} \tau(\psi) \, d\psi$.

The reason for this is that the set of states $\psi\in S_V$ with finite $\tau(\psi)$ has the structure of the unit sphere of a
subspace of $V$. This follows from the inequality $\|\bs{a}_n(\alpha\psi+\beta\phi)\|^2 \le
2|\alpha|^2\|\bs{a}_n\psi\|^2+2|\beta|^2\|\bs{a}_n\phi\|^2$, which gives $\tau(\alpha\psi+\beta\phi) \leq 2|\alpha|^2 \tau(\psi) + 2|\beta|^2
\tau(\phi)$. Therefore, unless this set exhausts all of $S_V$, it has null Lebesgue measure with respect to the uniform measure $d\psi$. Thus,
$\tau(\psi)$ diverges $d\psi$-a.e.~on $S_V$, which causes the divergence of the average $\dashint_{S_V} \tau(\psi) \, d\psi$.

The following theorem characterizes the subspaces $V$ with a finite averaged expected $V$-return time
and gives a formula for this average.
It can be considered as the extension to subspaces of the results given in \cite[Theorem 2]{GVWW}. A key ingredient will be the determinant
$\det\bs{T}$ of an operator $\bs{T}$ on $V$, that is, the determinant of any matrix representation of $\bs{T}$.

\begin{thm}
Consider a unitary step $U$ and a finite-dimensional subspace $V$ with spectral measure $\bs\mu(d\lambda)$, Schur function $\bs{f}(z)$ and first
$V$-return generating function $\hat{\bs{a}}(z)=z\bs{f}^\dag(z)$. Then, the following statements are equivalent:
\begin{enumerate}
\item All the states of $V$ are $V$-recurrent with a finite expected $V$-return time.
\item All the states of $V$ are recurrent with a finite expected return time.
\item $\bs\mu(d\lambda)$ is a sum of finitely many mass points.
\item $\bs{f}(z)$ is rational inner.
\item $\det\bs{f}(z)$ is rational inner.
\end{enumerate}
Under any of these conditions, the average of the expected $V$-return time is
\[
 \dashint_{S_V}\tau(\psi)\,d\psi = \frac{K}{\dim V}
\]
with $K$ a positive integer that can be computed equivalently as
\begin{equation}\label{eq:K:def}
 K = \sum_k \dim(E_kV) = \sum_k \rank\bs\mu(\{\lambda_k\}) = \deg\det\hat{\bs{a}}(e^{i\theta}),
\end{equation}
where $\lambda_k$ are the mass points of $\bs\mu(d\lambda)$, $E_k=E(\{\lambda_k\})$ are the orthogonal projectors onto the corresponding eigenspaces
of $U=\int\lambda\,E(d\lambda)$ and $\deg\det\hat{\bs{a}}(e^{i\theta})$ is the degree of $\det\hat{\bs{a}}(e^{i\theta}) \colon S^1 \to S^1$, i.e.~its
{\bf winding number}, which coincides with the number of the zeros of $\det\hat{\bs{a}}(z)$ inside the unit disk,
counting multiplicity.
\end{thm}

\begin{proof}
We know that $(3) \Leftrightarrow (4)$. On the other hand, the polarization identity implies that $\bs\mu(d\lambda)$ is finitely supported if and
only if the spectral measure $\mu_\psi=\<\psi|E(d\lambda)\psi\>$ of any state $\psi\in V$ is finitely supported too, so $(2) \Leftrightarrow (3)$
by \cite[Theorem 2]{GVWW}.
Hence, to prove the equivalences we only need to argue that $(1) \Rightarrow (5) \Rightarrow (4) \Rightarrow (1)$.
The formula $\dashint_{S_V}\tau(\psi)\,d\psi=\frac{K}{\dim V}$ will be established as a key step of the proof that $(4)\Rightarrow(1)$.

\medskip

\noindent \fbox{$(1) \Rightarrow (5)$}

The recurrence of $V$ means that $\bs{f}(z)$ is inner, i.e. it is unitary a.e.~on the unit circle. Therefore, $\det\bs{f}(z)$ is analytic on the unit
disk with unimodular boundary values, which implies that it is a scalar inner Schur function.

The finiteness of the average for the expected $V$-return time can be expressed by requiring that $\bs{f}(e^{i\theta})$ belongs to the space
$H^{1/2}$ of operator-valued functions $\bs{h}(\theta) = \sum_{n=-\infty}^\infty \bs{h}_n e^{in\theta}$ such that $\|\bs{h}\|_{H^{1/2}}^2 =
\sum_{n=-\infty}^\infty |n|\|\bs{h}_n\|_F^2 < \infty$. Since $\bs{f}(e^{i\theta}) \in L^\infty$, it follows that the matrix element
$\<\phi|\bs{f}(z)\psi\>$ lies in the scalar $H^{1/2} \cap L^\infty$ space for any states $\phi,\psi\in V$. Bearing in mind that $H^{1/2} \cap
L^\infty$ is an algebra, we conclude that $\det\bs{f}(e^{i\theta}) \in H^{1/2}$.

Hence, $\det \bs{f}(z)$ is a scalar inner Schur function in $H^{1/2}$. On the other hand, the one-dimensional analogue of this theorem, already
proved in \cite[Theorem 2]{GVWW}, implies that a scalar inner Schur function is in $H^{1/2}$ if and only if it is rational. Therefore, $\det
\bs{f}(z)$ is rational inner.

\medskip

\noindent \fbox{$(5) \Rightarrow (4)$}

As any rational inner Schur function is a finite Blaschke product,
\[
 \det\bs{f}(z) = \zeta \prod_{j=1}^d b_{w_j}(z), \qquad b_w(z)=\frac{z-w}{1-\overline wz}, \qquad |w|<1, \qquad |\zeta|=1,
\]
with finitely many zeros $w_1,\dots,w_d$ inside the unit disk.

We may factorize these zeros in the operator Schur function $\bs{f}(z)$ by using the rational inner Blaschke-Potapov factors \cite{DFK}
\[
 \qquad \bs{B}_{w,\phi}(z)=I_V+(b_w(z)-1)|\phi\>\<\phi|, \qquad \phi\in S_V.
\]
From Potapov's factorization, for some unit vectors $\phi_j\in\ker\bs{f}(w_j)$,
\[
 \bs{f}(z) = \bs{B}_{w_1,\phi_1}(z) \cdots \bs{B}_{w_d,\phi_d}(z) \, \bs{g}(z),
\]
with $\bs{g}(z)$ analytic on the unit disk. Since $\bs{g}(z)$ must be unitary on the unit circle too, it is an inner Schur function.

Bearing in mind that $\det\bs{B}_{w,\phi}(z)=b_w(z)$, we find that $\det\bs{g}(z)=\zeta$ is a unimodular constant for $|z|<1$. Consequently,
$\bs{g}(z)$ has an analytic inverse on the unit disk and, since the values are unitary on the boundary, not only $\|\bs{g}(z)\|\leq1$, but also
$\|\bs{g}(z)^{-1}\|\leq1$ for $|z|<1$. This implies that $\|\bs{g}(z)\psi\|=1$ for any $\psi\in S_V$ and $|z|<1$. Hence the polarization identity
and the maximum principle
show that $\bs{g}(z)\psi$ is constant on the unit disk. We conclude that $\bs{g}(z)$ is a constant unitary, therefore $\bs{f}(z)$ is rational inner.

\medskip

\noindent \fbox{$(4) \Rightarrow (1)$}

It suffices to derive the formula for $\dashint_{S_V}\tau(\psi)\,d\psi$ since $(1)$ is equivalent to the finiteness of this average.
Any rational inner Schur function $\bs{f}(z)$ is analytic in a neighborhood of the closed unit disk. Therefore, the restriction of
$\hat{\bs{a}}(z)=z\bs{f}^\dag(z)$ to the unit circle is a unitary function $\hat{\bs{a}}(e^{i\theta})$ that is real analytic in $\theta$. Hence, the
expression of the average of $\tau(\psi)$ in terms of radial limits to the unit circle becomes
\[
 \dashint_{S_V}\tau(\psi)\,d\psi = \frac{K}{\dim V},
 \qquad
 K = \int_0^{2\pi} \tr\left(\hat{\bs{a}}(e^{i\theta})^\dag \partial_\theta \hat{\bs{a}}(e^{i\theta})\right) \, \frac{d\theta}{2\pi i}.
\]

The extension of Rellich's theorem \cite{RELLICH,KATO} to real analytic families of normal operators ensures the existence of an orthonormal basis of
$V$ consisting of real analytic eigenvectors $\phi_k(\theta)$ of $\hat{\bs{a}}(e^{i\theta})$ with real analytic eigenvalues $\xi_k(\theta)$.

Differentiating $\hat{\bs{a}}(e^{i\theta}) = \sum_k \xi_k(\theta) |\phi_k(\theta)\>\<\phi_k(\theta)|$, we get
\begin{equation} \label{eq:DER-aa}
\begin{aligned}
 \bs{\hat{a}}^\dag \partial_\theta\bs{\hat{a}}
 = & \sum_k \overline{\xi_k} \partial_\theta\xi_k |\phi_k\>\<\phi_k|
 + \sum_{j \neq k}
 \left( \overline{\xi_j} \xi_k \<\phi_j|\partial_\theta\phi_k\>
 + \<\partial_\theta\phi_j|\phi_k\> \right)
 |\phi_j\>\<\phi_k|,
 \end{aligned}
\end{equation}
where we used $\re\<\phi_k(\theta)|\partial_\theta\phi_k(\theta)\>=0$ (since $\|\phi_k(\theta)\|=1$). This implies
\[
 \tr\left(\hat{\bs{a}}(e^{i\theta})^\dag \partial_\theta \hat{\bs{a}}(e^{i\theta})\right) =
 \sum_k\overline{\xi_k(\theta)}\partial_\theta\xi_k(\theta) =
 \overline{\det\hat{\bs{a}}(e^{i\theta})} \, \partial_\theta\det\hat{\bs{a}}(e^{i\theta}).
\]
Therefore,
\[
 K =
 \int_0^{2\pi} \frac{\partial_\theta\det\hat{\bs{a}}(e^{i\theta})}{\det\hat{\bs{a}}(e^{i\theta})} \, \frac{d\theta}{2\pi i} =
 \deg\det\hat{\bs{a}}(e^{i\theta})
\]
is the winding number of $\det\hat{\bs{a}}(e^{i\theta}) \colon S^1 \to S^1$, which is finite because it coincides with the number of zeros of the
finite Blaschke product $\det\hat{\bs{a}}(z)$.
The alternative expressions given in (\ref{eq:K:def}) for the positive integer $K$ are proved in Theorem \ref{thm:TAU-MASS} of Appendix
\ref{app:TAU-AVERAGE}.
\end{proof}

The previous theorem states that the presence of a finite-dimensional subspace $V$ whose states return to $V$ in a finite expected time is equivalent
to the existence of eigenvectors for the unitary step $U$. However, as mentioned previously, when such states do not cover all of $S_V$, no singular
subspace is required for $U$.
For instance, in $\ell^2(\mathbb{Z})$, $|0\>$ returns to $V=\spn\{|0\>,|1\>\}$ in one step with the unitary shift $U|x\>=|x+1\>$, $x\in\mathbb{Z}$.

The relation \eqref{eq:DER-aa} sheds light on the origin of the non-integer values of the geometrical phases
$\tau(\psi)$ and the reason for the cancellation of the non-integer contributions when averaging over the states $\psi \in S_V$. The identity
\[
 \dashint_{S_V} \<\psi|\eta\>\<\phi|\psi\>\,d\psi  = \frac{\<\phi|\eta\>}{\dim V}, \qquad \forall \phi,\eta\in\mathcal{H},
\]
shows that such an average cancels the off-diagonal terms in \eqref{eq:DER-aa} and transforms each projector $|\phi_k(\theta)\>\<\phi_k(\theta)|$
into the constant $1/\dim V$. This eliminates the dependence on $\phi_k(\theta)$ so that the only contributions to $\dashint_{S_V}\tau(\psi)\,d\psi$
come from the winding numbers of the eigenvalues $\xi_k(\theta)$, which are integers. However, if we do not perform the average, the transport of the
basis $\phi_k(\theta)$ along the closed path gives off-diagonal contributions, and also modifies the diagonal ones due to the $\theta$-dependence of
the projectors $|\phi_k(\theta)\>\<\phi_k(\theta)|$, causing the emergence of non-integer values. Nevertheless, $\tau(\psi)$ is independent of the
choice of basis functions $\phi_k(\theta)$ because the definition of the expected $V$-return time does not require such a choice. In other words,
while the expected $V$-return time $\tau(\psi)$ is a non-integer geometric property of the loop in $S_V$ defined by
$\psi(\theta)=\hat{\bs{a}}(e^{i\theta})\psi$, the average over $S_V$ of these geometrical phases leads to a true topological invariant, the winding
number of $\det\hat{\bs{a}}(e^{i\theta})$, which is therefore an integer.

Concerning the meaning of the requirement that $\bs\mu(d\lambda)$ is finitely supported, remember that $\bs\mu(d\lambda)=PE(d\lambda)P$, so the mass
points of $\bs\mu(d\lambda)$ must also be mass points of $E(d\lambda)$, which means that they are eigenvalues of $U$. If $\bs\mu(d\lambda)$ is a sum
of finitely many mass points $\lambda_k$, then $P = P \int E(d\lambda) P = P \sum_kE_k P$, which is equivalent to stating that $V \subset \oplus_k
E_kV$. Indeed, if $\psi\in V$ and $P=PQP$ with $Q=\sum_kE_k$, then $\psi=PQ\psi$. So $Q\psi=\psi+\phi$ with $\phi\in V^\perp$.  $\|Q\|=1$ implies
$\phi=0$. Conversely, if every $\psi\in V$ may be written $\psi=\sum_k E_k\psi_k$ with $\psi_k\in V$, then $Q\psi=\psi$ for all $\psi\in V$. Hence,
$PQP=P$. In summary, since $E_k$ is the orthogonal projector onto the eigenspace associated with $\lambda_k$, we find that the spectral measure of
$V$ is finitely supported if and only if $V$ is a subspace of a finite sum of eigenspaces of $U$.

The $U$-invariant subspace $\mathcal{H}_0=\oplus_k E_kV$ whose dimension gives $K$ has a remarkable meaning. Consider an arbitrary $U$-invariant
subspace $\hat{\mathcal{H}}$. The restriction $\hat{U}=U\upharpoonright\hat{\mathcal{H}}$ of $U$ to $\hat{\mathcal{H}}$ is a unitary with spectral
measure $\hat{E}(d\lambda)=E(d\lambda)\upharpoonright\hat{\mathcal{H}}$ and thus with eigenspaces $\hat{V}_k=E_k\hat{\mathcal{H}}$. Hence, if
$V\subset\hat{\mathcal{H}}$, then $E_kV\subset\hat{V}_k$, which implies that $\mathcal{H}_0\subset\hat{\mathcal{H}}$. This identifies $\mathcal{H}_0$
as the minimal $U$-invariant subspace in which $V$ lives, and allows us to state that {\it the average of the expected $V$-return time is the ratio
$\dim(\oplus_kE_kV)/\dim V$ between the dimension of the minimal $U$-invariant subspace containing $V$ and the dimension of $V$}.

In the case $\dim \mathcal{H} < \infty$ the integer $K$ can also be calculated as $\dim \mathcal{H} - \nu$, where $\nu$ is the number of linearly
independent eigenvectors of $U$ lying in $V^\bot$ (see Appendix \ref{app:TAU-AVERAGE}).

An additional consequence of the previous theorem is an identity that relates the degree of the determinant of a $D$-dimensional rational inner
Schur function $\bs{f}(z)$ and the masses of the corresponding measure $\bs\mu(d\lambda)$. Since $\hat{\bs{a}}(z)=z\bs{f}^\dag(z)$ we have that
$\det\hat{\bs{a}}(z) = z^D \overline{\det \bs{f}}(z)$, thus
\[
 \deg\det\bs{f}(e^{i\theta}) = \sum_k\rank\bs\mu(\{\lambda_k\}) - D,
\]
where $\lambda_k$ are the mass points of $\bs\mu(d\lambda)$.

\section{Applications: site recurrence} \label{sec:SITE}

The most natural example of subspace recurrence is what we can call site recurrence of a quantum walk. Consider a quantum walk on a spatial lattice
$\Lambda$ with a set of internal degrees of freedom $\Sigma$, i.e. the basis states $\phi_{x,\alpha}$ are labelled by the site index $x\in\Lambda$
running over the lattice and an additional index $\alpha\in\Sigma$ that refers to the possible basis states for each site. If we consider the
evolution starting from a state $\psi=\sum_{\alpha\in\Sigma}c_\alpha\phi_{x,\alpha}$ localized at a given site $x$, we can ask for the return
probability to the original site. In other words, we deal with the notion of $V_x$-recurrence for the subspace
$V_x=\mathrm{span}\{\phi_{x,\alpha}\}_{\alpha\in\Sigma}$.

In this section we will present explicit analytical results for the site recurrence of any coined quantum walk on the line, the half-line and a
finite one-dimensional lattice too. Some numerical calculations (given later) will show the effect of the dimension and the geometry of the lattice
on site recurrence.

\subsection{Site recurrence in one dimension} \label{ssec:SITE-1D}

Coined walks on the line live on the Hilbert space $\ell^2(\mathbb{Z})\otimes\mathbb{C}^2$, whose basis states we write as $|x,\alpha\>$,
$x\in\mathbb{Z}$, $\alpha\in\{\ua,\da\}$. The step operator $U = SC$ is factorized into a shift $S|x,\ua\> = |x+1,\ua\>$, $S|x,\da\> = |x-1,\da\>$
and a coin operation $C$ acting at each site subspace $V_x=\mathrm{span}\{|x,\ua\>,|x,\da\>\}$ separately as a two-dimensional unitary operator.
Bearing in mind the phase factor freedom of the basis states, we can assume without loss that the matrix representation in the basis
$\{|x,\ua\>,|x\da\>\}$ of the coin operator at site $x$ has the form
\[
 C_x = (\< x,\alpha| C |x,\beta\>)_{\alpha,\beta\in\{\ua,\da\}}
     = \begin{pmatrix} \rho_{2x} & -\gamma_{2x} \\ \overline\gamma_{2x} & \rho_{2x} \end{pmatrix},
 \quad
 \begin{aligned} & |\gamma_{2x}|\le1, \\ & \rho_{2x}=\sqrt{1-|\gamma_{2x}|^2}. \end{aligned}
\]
The convenience of using even indices for the coefficients of $C_x$ will become clear later on.

Actually, we will assume the strict inequality $|\gamma_{2x}|<1$. Otherwise the quantum walk would decouple trivially into independent ones living on
orthogonal subspaces generated by basis vectors that could be analyzed independently.

As we will show, the study of site recurrence for one-dimensional quantum walks is greatly aided by the theory of orthogonal polynomials on the unit
circle and CMV matrices; (for a very quick and modern introduction to both topics see \cite{SIMONfoot}). Indeed, they provide efficient tools for the
analytic computation of the $2 \times 2$ Schur functions that encode the site recurrence. Note that the polynomial approach becomes simpler for
quantum walks on the half-line or a finite one-dimensional lattice than for the case of the whole line \cite{CGMV1,CGMV2,CGMV3}. In the former case,
the Schur function is scalar-valued, while in the latter case, it is matrix-valued.

Our first step will be the study of site recurrence for coined walks on the half-line with reflecting boundary conditions at the origin. They are
defined as in the case of the whole line simply changing the lattice $\mathbb{Z}$ by $\mathbb{Z}_+=\{0,1,2,\dots\}$ and the shift at site 0 by
$S|0,\da\>=|0,\ua\>$. The case of a finite lattice $\mathbb{Z}_N=\{0,1,\dots,N-1\}$ with reflecting boundary conditions at the edges simply requires
an additional modification of the shift at the right edge by $S|N-1,\ua\>=|N-1,\da\>$.

\subsubsection{Coined walks on the half-line} \label{sssec:HL}

According to the previous paragraph, the transition matrix of a coined walk on the half-line has the form
\begin{equation} \label{eq:COIN-CMV-Z+}
 \mathcal{C} = (\< e_j|Ue_k\>)_{j,k\in\mathbb{Z}_+}
 = \text{\small
 $\begin{pmatrix}
 \overline{\gamma}_0 & \kern-3pt \rho_0 & \kern-3pt 0 & \kern-3pt 0 & \kern-3pt 0 & \kern-3pt 0 & \dots
 \\
 0 & \kern-3pt 0 & \kern-3pt \overline\gamma_2 & \kern-3pt \rho_2 & \kern-3pt 0 & \kern-3pt 0 & \dots
 \\
 \rho_0 & \kern-3pt -\gamma_0 & \kern-3pt 0 & \kern-3pt 0 & \kern-3pt 0 & \kern-3pt 0 & \dots
 \\
 0 & \kern-3pt 0 & \kern-3pt 0 & \kern-3pt 0 & \kern-3pt \overline\gamma_4 & \kern-3pt \rho_4 & \dots
 \\
 0 & \kern-3pt 0 & \kern-3pt \rho_2 & \kern-3pt -\gamma_2 & \kern-3pt 0 & \kern-3pt 0 & \dots
 \\
 0 & \kern-3pt 0 & \kern-3pt 0 & \kern-3pt 0 & \kern-3pt 0 & \kern-3pt 0 & \dots
 \\
 0 & \kern-3pt 0 & \kern-3pt 0 & \kern-3pt 0 & \kern-3pt \rho_4 & \kern-3pt -\gamma_4 & \dots
 \\
 \dots & \kern-3pt \dots & \kern-3pt \dots & \kern-3pt \dots & \kern-3pt \dots & \kern-3pt \dots & \dots
 \end{pmatrix}$},
\end{equation}
where we assume the ordering $e_{2x}=|x,\ua\>$, $e_{2x+1}=|x,\da\>$ for $x\in\mathbb{Z}_+$. This matrix is the special case $\gamma_{2k+1}=0$,
$k\in\mathbb{Z}_+$, of the more general unitary
\[
 \mathcal{C} = \text{\small $
 \begin{pmatrix}
 \overline{\gamma}_0 & \kern-5pt \rho_0 & \kern-5pt 0 & \kern-5pt 0 & \kern-5pt 0 & \kern-5pt 0 & \dots
 \\
 \rho_0\overline\gamma_1 & \kern-5pt -\gamma_0\overline{\gamma}_1 & \kern-5pt \rho_1\overline\gamma_2 & \kern-5pt \rho_1\rho_2 & \kern-5pt 0 &
 \kern-5pt 0 & \dots
 \\
 \rho_0\rho_1 & \kern-5pt -\gamma_0\rho_1 & \kern-5pt -\gamma_1\overline{\gamma}_2 & \kern-5pt -\gamma_1\rho_2 & \kern-5pt 0 & \kern-5pt 0 & \dots
 \\
 0 & \kern-5pt 0 & \kern-5pt \rho_2\overline\gamma_3 & \kern-5pt -\gamma_2\overline{\gamma}_3 & \kern-5pt \rho_3\overline\gamma_4 & \kern-5pt
 \rho_3\rho_4 & \dots
 \\
 0 & \kern-5pt 0 & \kern-5pt \rho_2\rho_3 & \kern-5pt -\gamma_2\rho_3 & \kern-5pt -\gamma_3\overline{\gamma}_4 & \kern-5pt -\gamma_3\rho_4 & \dots
 \\
 0 & \kern-5pt 0 & \kern-5pt 0 & \kern-5pt 0 & \kern-5pt \rho_4\overline\gamma_5 & \kern-5pt -\gamma_4\overline{\gamma}_5 & \dots
 \\
 0 & \kern-5pt 0 & \kern-5pt 0 & \kern-5pt 0 & \kern-5pt \rho_4\rho_5 & \kern-5pt -\gamma_4\rho_5 & \dots
 \\
 \dots & \kern-5pt \dots & \kern-5pt \dots & \kern-5pt \dots & \kern-5pt \dots & \kern-5pt \dots & \dots
 \end{pmatrix}$},
 \;
 \begin{aligned}
 & |\gamma_k|<1, \\ & \rho_k=\sqrt{1-|\gamma_k|^2},
 \end{aligned}
\]
known as a CMV matrix with Verblunsky parameters $(\gamma_k)$ \cite{FIVE,MINIMAL,SIMON1,SIMONfoot,SIMON5years,WATKINS}.

CMV matrices are closely related to the orthogonal Laurent polynomials on the unit circle, $X_k(z)$, which are defined by performing the Gram-Schmidt
procedure on the set $\{1,z,z^{-1},z^2,z^{-2},z^3,\dots\}$ in $L^2_\mu$ for a given measure $\mu(d\lambda)$ on $S^1$. They satisfy the recurrence
\[
 X(z) \, \mathcal{C} = z X(z),
 \quad
 X(z) = \begin{pmatrix} X_0(z) & \kern-5pt X_1(z) & \kern-5pt X_2(z) & \kern-5pt \dots \end{pmatrix},
 \quad
 X_0(z)=1.
\]

The CMV shape of the transition matrix provides a direct connection between coined walks on the half-line and orthogonal Laurent
polynomials on the unit circle with null odd Verblunsky parameters.  For more details, see \cite{CGMV1}.

The Verblunsky parameters of a CMV matrix have another important mathematical role: they characterize the Schur function $f(z)$ of the
orthogonality measure $\mu(d\lambda)$ via the Schur algorithm
\[
\begin{aligned}
 & f_0(z)=f(z),
 \\
 & f_{k+1}(z) = \frac{1}{z} \frac{f_k(z)-\gamma_k}{1-\overline\gamma_k f_k(z)}, \quad \gamma_k=f_k(0), \quad k\geq0,
\end{aligned}
\]
which generates a sequence $f_k(z)$ of Schur functions called the iterates of $f(z)$ \cite{SCHUR1,SCHUR2,KHRUSHCHEV,SIMON1,SIMONfoot}. Geronimus'
theorem \cite{SIMON1} implies that the Verblunsky parameters $(\gamma_k)$ in the CMV matrix are equal to the Schur parameters $(\gamma_k)$ in Schur's
algorithm. Note that the $k$-th iterate $f_k(z)$ has Schur parameters $(\gamma_k,\gamma_{k+1},\gamma_{k+2},\dots)$.

Coined walks on the half-line correspond to Schur parameters $(\gamma_0,0,\gamma_2,0,\gamma_4,0,\dots)$. From the Schur algorithm, we see that the
absence of odd terms is equivalent to the Schur function $f(z)$ being an even function of $z$.

The orthogonality measure of a CMV matrix has an infinite support iff there are infinitely many Schur parameters. If the measure is supported on a
finite number of points, say $N$, the corresponding Schur function will have the same finite number of Schur parameters
$(\gamma_0,\gamma_1,\dots,\gamma_{N-1})$, all of them in the open unit disk except the last one, which lies on the unit circle. The reason for this
is that the Schur function is a Blaschke product of degree $N-1$, i.e. a rational inner function, and thus the Schur algorithm stops after $N-1$
steps because it reaches the unimodular constant function $\gamma_{N-1}$. Therefore, an inner Schur function is rational exactly when it has finitely
many Schur parameters.

Together with the standard iterates, the inverse iterates are key in the study of recurrence. If the Schur function $f(z)$ has Schur parameters
$(\gamma_0,\gamma_1,\dots)$, the rational inner Schur function $f^k(z)$ with Schur parameters
$(-\overline\gamma_k,-\overline\gamma_{k-1},\dots,-\overline\gamma_0,1)$ is called the $k$-th inverse iterate of $f(z)$. This terminology was
introduced in \cite{KHRUSHCHEV2}, where the author uses the notation $b_{k-1}(z)$ for the function we denote $f^k(z)$ for convenience.

An important consequence of the above results is that the step operator
$U$ of a coined walk on the half-line is unitarily equivalent to the
multiplication operator on $L^2_\mu$ given by $h(z) \mapsto zh(z)$. The equivalence is established by $e_k \mapsto X_k(z)$, so $\< e_j|U^ne_k \> =
\int \overline{X_j(\lambda)} \, \lambda^n X_k(\lambda) \, \mu(d\lambda)$ for $n\in\mathbb{Z}$. On the other hand, the spectral decomposition $U=\int
\lambda\,E(d\lambda)$ yields $\< e_j|U^ne_k \> = \int \lambda^n \< e_j|E(d\lambda)e_k \>$. This identifies the complex-valued measure
\[
  \<e_j|E(d\lambda)e_k \> = \overline{X_j(\lambda)} \, X_k(\lambda) \, \mu(d\lambda)
\]
in terms of the orthogonal Laurent polynomials and the corresponding orthogonality measure. Hence, the spectral measure
$\mu_\psi(d\lambda)=\<\psi|E(d\lambda)\psi\>$ of any state $\psi=\sum_{k\geq0}c_ke_k$ has the form $\mu_\psi(d\lambda)=|h(\lambda)|^2\mu(d\lambda)$
with $h(z)=\sum_{k\geq0}c_kX_k(z) \in L^2_\mu$, so it is absolutely continuous with respect to $\mu(d\lambda)$. In particular, the orthogonality
measure $\mu(d\lambda)$ coincides with the spectral measure of the first basis state $e_0=|0,\ua\>$.

Concerning site recurrence, identifying the spectral measure $\bs\mu_x(d\lambda)$ of the site subspace $V_x=\mathrm{span}\{e_{2x},e_{2x+1}\}$ with
its matrix representation in the basis $\{e_{2x},e_{2x+1}\}$, we can write
\begin{equation}\label{eq:mux:half:line}
\begin{aligned}
 \bs\mu_x(d\lambda) & = (\< e_j|E(d\lambda)e_k \>)_{j,k\in\{2x,2x+1\}}
 \\
 & = \begin{pmatrix}
 |X_{2x}(\lambda)|^2 & \overline{X_{2x}(\lambda)} X_{2x+1}(\lambda)
 \\
 \overline{X_{2x+1}(\lambda)} X_{2x}(\lambda) & |X_{2x+1}(\lambda)|^2
 \end{pmatrix} \mu(d\lambda).
\end{aligned}
\end{equation}
The corresponding matrix Schur function, obtained in Appendix \ref{app:SITE-SCHUR}, is
\begin{equation} \label{eq:fx:half:line}
 \bs{f}_{\kern-2pt x}(z)
 = \begin{pmatrix} \gamma_{2x} f^{2x-1}(z) & \rho_{2x} f_{2x+1}(z) \\ \rho_{2x} f^{2x-1}(z) & -\overline\gamma_{2x} f_{2x+1}(z) \end{pmatrix},
\end{equation}
where $f_k(z)$ and $f^k(z)$ are the iterates and inverse iterates of the Schur function $f(z)$ corresponding to the orthogonality measure
$\mu(d\lambda)$. More precisely, the Schur parameters of $f_{2x+1}(z)$ and $f^{2x-1}(z)$ are
\[
  (0,\gamma_{2x+2},0,\gamma_{2x+4},\dots)
  \qquad \text{and} \qquad
  (0,-\overline\gamma_{2x-2},0,-\overline\gamma_{2x-4},\dots,-\overline\gamma_0,1),
\]
respectively. In particular, $f^{2x-1}(z)$ is a finite Blaschke
product of degree $2x$.

To obtain the site return probability operator we must simply take into account that $f^k(z)$ is an inner function. Then, the previous matrix Schur
function gives the $V_x$-return probability matrix
\[
\begin{aligned}
 \bs{R}_x & = \int_0^{2\pi} \bs{f}_{\kern-2pt x}(e^{i\theta}) \, \bs{f}_{\kern-2pt x}(e^{i\theta})^\dag \, \frac{d\theta}{2\pi}
 \\
 & = \begin{pmatrix}
 |\gamma_{2x}|^2+\rho_{2x}^2\|f_{2x+1}\|^2 & \rho_{2x}\gamma_{2x}(1-\|f_{2x+1}\|^2)
 \\
 \rho_{2x}\overline\gamma_{2x}(1-\|f_{2x+1}\|^2) & \rho_{2x}^2+|\gamma_{2x}|^2\|f_{2x+1}\|^2
 \end{pmatrix}
 \\
 & = I_2 - (1-\|f_{2x+1}\|^2) \bs\Gamma_x, \quad
 \bs\Gamma_x = \begin{pmatrix} \rho_{2x}^2 & -\rho_{2x}\gamma_{2x} \\ -\rho_{2x}\overline\gamma_{2x} & |\gamma_{2x}|^2 \end{pmatrix},
\end{aligned}
\]
where $I_2$ is the $2 \times 2$ identity matrix and $\|\cdot\|$ is the $L^2$ norm with respect to Lebesgue measure on the unit circle.

The matrix $\bs\Gamma_x$ has eigenvectors $(\gamma_{2x},\rho_{2x})^T$ and $(\rho_{2x},-\overline\gamma_{2x})^T$ with eigenvalues 0 and 1
respectively. Therefore, the maximum and minimum return probabilities to the site $x$ are $1$ and $\|f_{2x+1}\|^2$, attained respectively at
the qubits
\begin{equation} \label{eq:EXTREME}
 \psi_x^{(1)}=\gamma_{2x}|x,\ua\>+\rho_{2x}|x,\da\>, \qquad \psi_x^{(2)}=\rho_{2x}|x,\ua\>-\overline\gamma_{2x}|x,\da\>.
\end{equation}
The average of the return probability to the site $x$ is therefore
\[
 \dashint_{S_{V_x}} \<\psi|\bs{R}_x\psi\> \, d\psi = \frac{\mathrm{Tr}\bs{R}_x}{\dim V_x} = \frac{1+\|f_{2x+1}\|^2}{2}.
\]

From the above results, we find that any site of a coined walk on the half-line has at least one site-recurrent qubit,
i.e.~a qubit that returns to
the site with probability one. The orthogonality measure $\mu(d\lambda)$ of a CMV transition matrix
like \eqref{eq:COIN-CMV-Z+} can be absolutely continuous,
from which it follows that the spectral measure $\mu_\psi(d\lambda)$ of any state is also absolutely continuous.
Therefore, coined walks on the half-line associated with absolutely continuous orthogonality measures $\mu(d\lambda)$ are examples of quantum walks
without a singular subspace for the unitary step $U$, but having a $V_x$-recurrent qubit for any site subspace $V_x$.

Furthermore, if a site is recurrent then all the states of the Hilbert space must be recurrent, which implies that all the sites are simultaneously
recurrent. Indeed, the previous results show that $V_x$ is a recurrent subspace if and only if $f_{2x+1}(z)$ is inner. Since a Schur function $f(z)$
and its iterates $f_k(z)$ are simultaneously inner or non-inner, we conclude that the presence of a recurrent site implies that the orthogonality
measure $\mu(d\lambda)$ is purely singular. Then, not only does every qubit become site recurrent, but every state becomes recurrent as well since
its spectral measure is absolutely continuous with respect to $\mu(d\lambda)$.

A remarkable fact is that the return probability matrix $\bs{R}_x$ for the site $x$ only depends on the Schur parameters
$(\gamma_{2x},\gamma_{2x+2},\gamma_{2x+4},\dots)$. This means that the site return probability of a qubit at a site $x$ does not depend on the coins
$C_y$, $y<x$, of the previous sites. This is the site extension of a result already found in \cite[Section 6.1]{GVWW} for the state recurrence of
coined walks on the half-line.

Concerning the expected return time to the site $x$, we should distinguish two cases depending on whether the measure $\mu(d\lambda)$ is
purely singular or not. If not, the expected $V_x$-return time makes sense (i.e.~it is not necessarily infinite) only for
$\psi_x^{(1)}$ because it is the unique site-recurrent qubit at site $x$. We find that $\tau(\psi_x^{(1)}) = 1 + \delta\tau(\psi_x^{(1)}) = 2x+1$ is
an integer because
\[
\begin{aligned}
 \delta\tau(\psi_x^{(1)}) & = \lim_{r\to1^-} \int_0^{2\pi}
 \begin{pmatrix} \overline\gamma_{2x} & \rho_{2x} \end{pmatrix}
 \partial_\theta\bs{f}_{\kern-2pt x}(re^{i\theta})\bs{f}_{\kern-2pt x}(re^{i\theta})^\dag
 \begin{pmatrix} \gamma_{2x} \\ \rho_{2x} \end{pmatrix}
 \, \frac{d\theta}{2\pi i}
 \\
 & = \int_0^{2\pi} \frac{\partial_\theta f^{2x-1}(e^{i\theta})}{f^{2x-1}(e^{i\theta})} \, \frac{d\theta}{2\pi i} = \deg f^{2x-1},
\end{aligned}
\]
i.e. $\delta\tau(\psi_x^{(1)})$ is the winding number of the finite Blaschke product $f^{2x-1}(e^{i\theta})$ of degree $2x$.

If $\mu(d\lambda)$ is purely singular, then all the sites are recurrent and the expected return time to the site $x$ makes sense for any qubit
$\psi_x = \alpha|x,\ua\>+\beta|x,\da\>$. Just as in the previous case we obtain
\[
 \tau(\psi_x) =
 1
 + |\alpha\overline\gamma_{2x}+\beta\rho_{2x}|^2 \deg f^{2x-1}
 + \lim_{r\to1^-}|\alpha\rho_{2x}-\beta\gamma_{2x}|^2 \, \tau_r(f_{2x+1}),
\]
where, for any Schur function $g(z)=\sum_{n\ge0}c_nz^n$ and $r\in(0,1)$, we use the notation
\begin{equation}\label{eq:tau:r}
 \tau_r(g) = \int_0^{2\pi} \overline{g(re^{i\theta})} \, \partial_\theta g(re^{i\theta}) \, \frac{d\theta}{2\pi i} = \sum_{n\ge1} n|c_n|^2r^{2n}.
\end{equation}
When $g(z)$ is rational inner, obviously $\lim_{r\to1^-} \tau_r(g) = \deg g$. On the other hand, if $g(z)$ is inner but not rational -- i.e. it is
inner with an infinite sequence of Schur parameters -- it is proved in \cite[Theorem 2]{GVWW} that $\lim_{r\to1^-} \tau_r(g)=\infty$. This last one is
the situation for the Schur function $f(z)$ of $\mu(d\lambda)$ and its iterates. Hence, we finally get
\[
 \tau(\psi_x) = \begin{cases} 2x+1, & \mathrm{if} \; \psi_x\in\mathrm{span}\{\psi_x^{(1)}\}, \\ \infty, & \mathrm{otherwise.} \end{cases}
\]
As a consequence, $\dashint_{S_{V_x}}\tau(\psi)\,d\psi=\infty$, which is in agreement with the fact that the spectral measure $\bs\mu_x(d\lambda)$
in (\ref{eq:mux:half:line}) is never finitely supported because the orthogonality measure $\mu(d\lambda)$ has an infinite support.

These results prove that for any site $x$, coined walks on the half-line always have a $V_x$-recurrent state with a finite expected $V_x$-return
time, independently of the presence or absence of a singular subspace for the unitary step $U$.

\subsubsection{Coined walks on a finite lattice} \label{sssec:FL}

The ordering $e_0=|0,\ua\>,e_1=|0,\da\>,\dots,e_{2N-2}=|N-1,\ua\>,e_{2N-1}=|N-1,\da\>$ of the basis states yields the following transition matrix for
a coined walk on the finite lattice $\mathbb{Z}_N$,
\[
 \mathcal{C} =
 \text{\small
 $\begin{pmatrix}
 \overline{\gamma}_0 & \rho_0
 \\
 0 & 0 & \overline\gamma_2 & \rho_2
 \\
 \rho_0 & -\gamma_0 & 0 & 0
 \\
 && 0 & 0
 \\
 && \rho_2 & -\gamma_2
 \\
 &&&& \ddots
 \\
 &&&&& \overline\gamma_{2N-4} & \kern-5pt \rho_{2N-4}
 \\
 &&&&& 0 & \kern-5pt 0
 \\
 &&&&& 0 & \kern-5pt 0 & \kern-5pt \overline\gamma_{2N-2} & \kern-5pt \rho_{2N-2}
 \\
 &&&&& \rho_{2N-4} & \kern-5pt -\gamma_{2N-2} & \kern-5pt 0 & \kern-5pt 0
 \\
 &&&&&&& \kern-5pt \rho_{2N-2} & \kern-5pt -\gamma_{2N-2}
\end{pmatrix}$}.
\]
This transition matrix is a finite CMV matrix with Verblunsky parameters
\[
  (\gamma_0,0,\gamma_2,0,\dots,\gamma_{2N-2},1).
\]
Unlike the case of a semi-infinite CMV matrix, this one generates only a finite segment of $2N$ orthogonal Laurent polynomials.
The orthogonality measure $\mu(d\lambda)$
is again the spectral measure of the first basis state $e_0$, but now it is supported on $2N$ points.
Consequently, for a coined walk on a finite one-dimensional lattice, every subspace is recurrent with a finite expected return time for any state.

Concerning site recurrence, it remains to calculate the expected return time for any qubit. For this purpose we need the Schur function
$\bs{f}_{\kern-2pt x}(z)$ of a site $x$, which is the same one already given for the half-line, but taking into account that the iterate
$f_{2x+1}(z)$ of the Schur function $f(z)$ for $\mu(d\lambda)$ now has a finite segment of Schur parameters
$(0,\gamma_{2x+2},0,\gamma_{2x+4},\dots,0,\gamma_{2N-2},1),$ and thus it is a finite Blaschke product of degree $2(N-1-x)$. We conclude that the
expected $V_x$-return time for a qubit $\psi_x=\alpha|x,\ua\>+\beta|x,\da\>$ is given by the quadratic form
\[
 \tau(\psi_x) =
 \begin{pmatrix} \overline\alpha & \overline\beta \end{pmatrix}
 \bs\tau_{\kern-1pt x}
 \begin{pmatrix} \alpha \\ \beta \end{pmatrix},
\]
with $\bs\tau_{\kern-1pt x} = I_2 + \bs{\delta\tau}_{\kern-1pt x}$ and
\[
\begin{aligned}
 \bs{\delta\tau}_{\kern-2pt x}
 & = \int_0^{2\pi} \partial_\theta\bs{f}_{\kern-2pt x}(e^{i\theta}) \bs{f}_{\kern-2pt x}(e^{i\theta})^\dag \, \frac{d\theta}{2\pi}
 \\
 & = \begin{pmatrix}
 |\gamma_{2x}|^2 \deg f^{2x-1} + \rho_{2x}^2 \deg f_{2x+1} & \rho_{2x}\gamma_{2x} (\deg f^{2x-1} - \deg f_{2x+1})
 \\
 \rho_{2x}\overline\gamma_{2x} (\deg f^{2x-1} - \deg f_{2x+1}) & \rho_{2x}^2 \deg f^{2x-1} + |\gamma_{2x}|^2 \deg f_{2x+1}
 \end{pmatrix}
 \\
 & = \deg f^{2x-1} I_2 + (\deg f_{2x+1} - \deg f^{2x-1}) \bs\Gamma_x.
\end{aligned}
\]

The extreme values of the expected return time for the site $x$, attained at the qubits $\psi_x^{(1)}$ and $\psi_x^{(2)}$ given in
\eqref{eq:EXTREME}, are the integers $1 + \deg f^{2x-1} = 2x+1$ and $1 + \deg f_{2x+1} = 2(N-1-x)+1$, which do not depend on the coins but only on
the number of sites to the left and the right of the site $x$ respectively. However, there exist qubits at site $x$ whose expected return time to
such a site is any value between the above extreme integers, which shows the non-integer nature of the Aharonov-Anandan phase giving the expected
$V_x$-return time. Despite this fact, the average of $\tau(\psi_x)$ is, not only an integer,
\[
 \dashint_{S_{V_x}} \tau(\psi) \, d\psi
 = \frac{\mathrm{Tr} \, \bs\tau_{\kern-1pt x}}{\dim V_x} = \frac{(1 + \deg f^{2x-1}) + (1 + \deg f_{2x+1})}{2} = N,
\]
but also independent of the site $x$ and the coins of the walk: it is always equal to the number of sites of the lattice.

The above average can be calculated also from the degree of
\[
 \det\hat{\bs{a}}_x(z)=\det{z\bs{f}_{\kern-2pt x}^\dag}(z)=-z^2\overline{f^{2x-1}}(z)\overline{f_{2x+1}}(z),
\]
which gives
\[
 \dashint_{S_{V_x}} \tau(\psi) \, d\psi = \frac{\deg\det\hat{\bs{a}}_x}{\dim V_x} = \frac{2 + \deg f^{2x-1} + \deg f_{2x+1}}{2} = N.
\]

\subsubsection{Coined walks on the line} \label{sssec:L}

The transition matrix with respect to the basis ordered as $e_{2x}=|x,\ua\>$, $e_{2x+1}=|x,\da\>$, $x\in\mathbb{Z}$, is now a doubly-infinite CMV
matrix with Verblunsky parameters $(\dots,0,\gamma_{-2},0,\gamma_0,0,\gamma_2,0,\dots)$ \cite{CGMV1,CGMV2,CGMV3}.

As in the previous cases, the analysis of site recurrence for the whole line comes from the Schur function of a site. Appendix \ref{app:SITE-SCHUR}
shows that this Schur function can be read from the case of the half-line by extending to a doubly-infinite sequence of Schur parameters
$(\gamma_k)_{k\in\mathbb{Z}}$ the definitions of $f_k(z)$ and $f^k(z)$ as the Schur functions with infinitely many Schur parameters
$(\gamma_k,\gamma_{k+1},\gamma_{k+2},\dots)$ and $(-\overline\gamma_k,-\overline\gamma_{k-1},-\overline\gamma_{k-2},\dots)$, respectively. In other
words, the expression for $\bs{f}_{\kern-2pt x}(z)$ given in \eqref{eq:fx:half:line} for the half-line also holds for the whole line with
$f_{2x+1}(z)$ and $f^{2x-1}(z)$ corresponding to the Schur parameters $(0,\gamma_{2x+2},0,\gamma_{2x+4},0,\dots)$ and
$(0,-\overline\gamma_{2x-2},0,-\overline\gamma_{2x-4},\dots)$ respectively.

The main difference compared to the half-line case is that $f^{2x-1}(z)$ is not a rational inner function on the line.
Therefore, the return probability matrix to the site $x$ takes the form
\[
\begin{aligned}
 \bs{R}_x
 & = \begin{pmatrix}
 |\gamma_{2x}|^2\|f^{2x-1}\|^2+\rho_{2x}^2\|f_{2x+1}\|^2 & \rho_{2x}\gamma_{2x}(\|f^{2x-1}\|^2-\|f_{2x+1}\|^2)
 \\
 \rho_{2x}\overline\gamma_{2x}(\|f^{2x-1}\|^2-\|f_{2x+1}\|^2) & \rho_{2x}^2\|f^{2x-1}\|^2+|\gamma_{2x}|^2\|f_{2x+1}\|^2
 \end{pmatrix}
 \\
 & = \|f^{2x-1}\|^2 I_2 + (\|f_{2x+1}\|-\|f^{2x-1}\|^2) \bs\Gamma_x.
\end{aligned}
\]
The extreme values of the $V_x$-return probability are $\|f^{2x-1}\|^2$, $\|f_{2x+1}\|^2$ and they are again reached at the qubits $\psi_x^{(1)}$,
$\psi_x^{(2)}$ given in \eqref{eq:EXTREME}. The average over $V_x$ of the return probability to the site $x$ is
\[
 \dashint_{S_{V_x}} \<\psi|\bs{R}_x\psi\> \, d\psi = \frac{\|f^{2x-1}\|^2+\|f_{2x+1}\|^2}{2}.
\]

Site-recurrent qubits appear if and only if $f^{2x-1}(z)$ or $f_{2x+1}(z)$ are inner for some site $x$, and a site $x$ is recurrent exactly when
$f^{2x-1}(z)$ and $f_{2x+1}(z)$ are both inner. However, the fact that $f_{k+j}(z)$ and $f^{k-j}(z)$ are the $j$-th iterates of $f_k(z)$ and $f^k(z)$
respectively implies that the Schur functions $f_{2x+1}(z)$ are simultaneously inner or non-inner for all $x$, and analogously for $f^{2x-1}(z)$.
Hence, concerning site recurrence, all the sites exhibit the same characteristics:
either none of them has a site-recurrent qubit; or all of them have a single site-recurrent qubit; or all of them are recurrent sites.

Further, as on the half-line, the presence of a recurrent site implies not only the recurrence of any other site, but the recurrence of any state,
and hence of any subspace. The reason for this is that the spectral matrix measure $\bs\mu(d\lambda)$ of the subspace spanned by the basis states
$|0,\ua\>,\left|-1,\da\>\right.$ generates the measure of any other state through an expression of the form
$\bs{h}(\lambda)^\dag\bs\mu(d\lambda)\bs{h}(\lambda)$ for some vector function $\bs{h}(\lambda)\in L^2_\mu$, while the matrix Schur function of
$\bs\mu(d\lambda)$ is \cite{CGMV2,CGMV3}
\[
 \bs{f}(z) = \begin{pmatrix} 0 & f_0(z) \\ f^{-1}(z) & 0 \end{pmatrix}.
\]
Obviously $f_{2x+1}(z)$ and $f_0(z)$ are simultaneously inner or non-inner, as are $f^{2x-1}(z)$ and $f^{-1}(z)$. Hence, the presence of recurrent
sites is equivalent to $f_0(z)$ and $f^{-1}(z)$ both being inner, which at the same time is equivalent to
$\bs{f}(z)$ being inner. This means
that the matrix measure $\bs\mu(d\lambda)$ -- and thus the spectral measure of any state -- is purely singular, implying the recurrence of
any state.

A special situation holds when $\|f^{2x-1}\|=\|f_{2x+1}\|$, a condition that characterizes the sites $x$ with an effectively scalar site
return probability matrix
$\bs{R}_x=\|f_{2x+1}\|^2I_2$, i.e.~with the same site return probability for any qubit.
This holds for instance if the Schur parameters are
symmetric ($\gamma_{2x-2k}=e^{i(\eta+2k\omega)}\gamma_{2x+2k}$) or conjugated ($\gamma_{2x-2k}=e^{i(\eta+2k\omega)}\overline\gamma_{2x+2k}$)
with respect to the site $x$, up to a linear phase. In the former case, $f^{2x-1}(z)=-e^{-i\eta}\overline{f_{2x+1}}(e^{-i\omega}z)$, and in the
latter, $f^{2x-1}(z)=-e^{-i\eta}f_{2x+1}(e^{-i\omega}z)$, as follows from the Schur algorithm. In particular, for a constant coin on the line,
all the qubits of all the sites have the same site return probability since $f_{2x+1}(z)$ and $f^{2x-1}(z)=-\overline{f_{2x+1}}(z)$ are independent
of $x$.

Given a site-recurrent qubit $\psi_x=\alpha|x,\ua\>+\beta|x,\da\>$, its expected return time to the site $x$ is given by
\[
\begin{aligned}
 \tau(\psi_x) & = \lim_{r\to1^-}
 \begin{pmatrix} \overline\alpha & \overline\beta \end{pmatrix}
 \bs\tau_{\kern-1pt x,r}
 \begin{pmatrix} \alpha \\ \beta \end{pmatrix},
 \qquad
 \bs\tau_{\kern-1pt x,r} = I_2 + \bs{\delta\tau}_{\kern-1pt x,r},
 \\
 \bs{\delta\tau}_{\kern-1pt x,r}
 & = \int_0^{2\pi} \partial_\theta\bs{f}_{\kern-2pt x}(re^{i\theta}) \bs{f}_{\kern-2pt x}(re^{i\theta})^\dag \, \frac{d\theta}{2\pi}
 \\
 & = \begin{pmatrix}
 |\gamma_{2x}|^2 \, \tau_r(f^{2x-1}) + \rho_{2x}^2 \, \tau_r(f_{2x+1}) & \rho_{2x}\gamma_{2x} (\tau_r(f^{2x-1}) - \tau_r(f_{2x+1}))
 \\
 \rho_{2x}\overline\gamma_{2x} (\tau_r(f^{2x-1}) - \tau_r(f_{2x+1})) & \rho_{2x}^2 \, \tau_r(f^{2x-1}) + |\gamma_{2x}|^2 \, \tau_r(f_{2x+1})
 \end{pmatrix}
 \\
 & = \tau_r(f^{2x-1}) \, \bs\Gamma^x + \tau_r(f_{2x+1}) \, \bs\Gamma_x,
 \\
 & \bs\Gamma^x = \begin{pmatrix} |\gamma_{2x}|^2 & \rho_{2x}\gamma_{2x} \\ \rho_{2x}\overline\gamma_{2x} & \rho_{2x}^2 \end{pmatrix},
 \quad
 \bs\Gamma_x = \begin{pmatrix} \rho_{2x}^2 & -\rho_{2x}\gamma_{2x} \\ -\rho_{2x}\overline\gamma_{2x} & |\gamma_{2x}|^2 \end{pmatrix}.
\end{aligned}
\]
where $\tau_r$ was defined in (\ref{eq:tau:r}) above.
The matrices $\bs\Gamma_x$ and $\bs\Gamma^x$ are positive semidefinite and their kernels -- spanned by $\psi_x^{(1)}$ and $\psi_x^{(2)}$ respectively
-- are orthogonal to each other. Note that $\tau_r(\cdot)\ge0$, and $\lim_{r\to1^-} \tau_r(f_k) = \infty$ whenever $f_k(z)$
is inner because it has infinitely many Schur parameters, with a similar result for $f^k(z)$. As a consequence,
$\tau(\psi_x)=\infty$ for any site-recurrent qubit $\psi_x$. Thus, in
contrast to the case of the half-line, the expected return time to a site is infinite for any site-recurrent qubit on the whole line.

\medskip

We have seen that the analysis of site recurrence for any coined walk on a one-dimensional lattice can be reduced to the study of scalar Schur
functions, despite the fact that subspace recurrence is characterized by matrix Schur functions. This is due to the form of the site Schur function,
which is given in terms of `right' and `left' scalar Schur functions characterized by the Schur parameters of the coins located to the right and left
of the site under consideration, respectively. Moreover, a reflecting boundary condition at a left or right edge implies that the corresponding
scalar Schur function has a finite number of Schur parameters, and thus is rational inner. This results in the independence of the site return
probability with respect to the coins on the side of the site where the reflecting boundary condition appears. In any case, the return probability
matrix to a site only depends on the coin at that site and the $L^2$ norm of the right and left Schur functions.

These norms can be computed in some cases, providing a completely explicit expression for the site return probability matrix.
For instance, this is the case for a constant coin
\begin{equation} \label{eq:CC}
 C = \begin{pmatrix} \rho & -\gamma \\ \overline\gamma & \rho \end{pmatrix}, \qquad \rho=\sqrt{1-|\gamma|^2},
\end{equation}
which yields \cite{GVWW,CGMV3}
\[
 \|f_{2x+1}\|^2 = \|f^{2x-1}\|^2=\frac{2}{\pi|\gamma|^2}\left\{\rho|\gamma|+(1-2\rho^2)\arcsin|\gamma|\right\}
\]
for any site $x$, assuming no reflecting boundary conditions at the right or left edge, respectively; (otherwise
the related $L^2$ norm is one).

Finally, from the Taylor expansion around the origin of $f_{2x+1}(z)$ and $f^{2x-1}(z)$, we can obtain the Taylor coefficients of the site Schur
function $\bs{f}_{\kern-2pt x}(z)=\sum_{n\ge1}\bs{a}_{n,x}^\dag z^{n-1}$, which provide the $n$-step first return amplitude matrix $\bs{a}_{n,x}$ to
the site $x$. For instance, the constant coin \eqref{eq:CC} on the line yields $f_{2x+1}(z)=\sum_{n\ge1}\overline{c}_nz^n$ and
$f^{2x+1}(z)=-\sum_{n\ge1}c_nz^n$ with \cite{GVWW,CGMV3}
\[
 c_1=\overline\gamma; \quad c_{2n}=0, \quad c_{2n+1}=\frac{P_{n-1}(c)-cP_n(c)}{2\gamma(n+1)}, \quad c=1-2|\gamma|^2, \quad n\ge1,
\]
and $P_n(z)$ the Legendre polynomials. Hence,
\[
\begin{array}{c}
 \bs{a}_{x,2n-1}=0, \qquad \bs{a}_{x,2n} = d_n \Upsilon, \qquad n\ge1,
 \smallskip \\
 d_n = \begin{cases} 1, & n=1, \\ \frac{P_{n-2}(c)-cP_{n-1}(c)}{2|\gamma|^2n}, & n\ge2, \end{cases}
 \qquad
 \Upsilon = \begin{pmatrix} -|\gamma|^2 & -\rho\gamma \\ \rho\overline\gamma & -|\gamma|^2 \end{pmatrix}.
\end{array}
\]

Therefore, given two qubits $\psi_x,\phi_x$ at site $x$, the $n$-step amplitude for the transition $\psi_x \to \phi_x$ without hitting the site $x$
at intermediate steps is $d_n\<\phi_x|\Upsilon\psi_x\>$. The time and qubit dependence of these amplitudes enter through the scalar and matrix
factors $d_n$ and $\Upsilon$, respectively. $\Upsilon$ encodes the ratio between the $n$-step amplitudes (without hitting the site earlier) for
different pairs of qubits, which are independent of $n$ and the site. The time dependence, isolated in the scalar factor $d_n$, shows an $n^{-3/2}$
decay as $n\to\infty$ \cite{GVWW,CGMV3}.

The comparison between state and site recurrence for coined walks leads to interesting phenomena seen previously in
Example~\ref{ex:SHIFT+EIG-Schur}.  We can compute the state return probability of a qubit following steps similar to those in
Example~\ref{ex:SHIFT+EIG-Schur}. Indeed,
using the matrix Carath\'{e}odory function $\bs{F}_x(z)$ corresponding to the site $x$, obtained from the site Schur
function $\bs{f}_x(z)$, we can arrive at the scalar Carath\'{e}odory function
$\left(\begin{smallmatrix}\overline\alpha&\overline\beta\end{smallmatrix}\right) \bs{F}_x(z)
\left(\begin{smallmatrix}\alpha\cr\beta\end{smallmatrix}\right)$ of any qubit $\psi_x = \alpha|x,\ua\>+\beta|x,\da\>$. The Lebesgue $L^2$ norm of the
related scalar Schur function yields the return probability to the same qubit.

\begin{figure}
 \includegraphics[width=12cm,height=5cm]{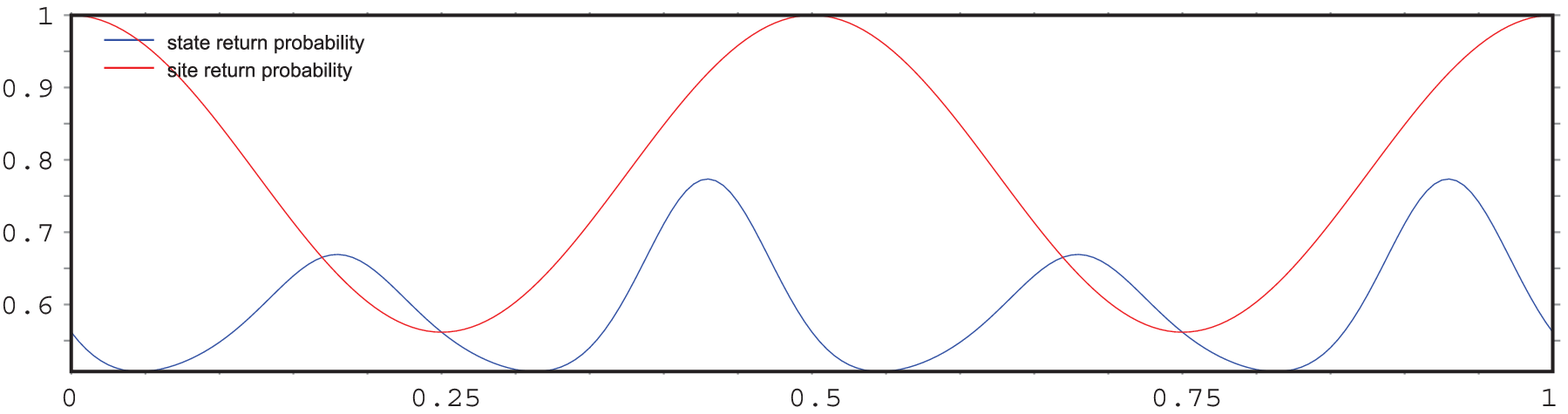}
 \includegraphics[width=12cm,height=5cm]{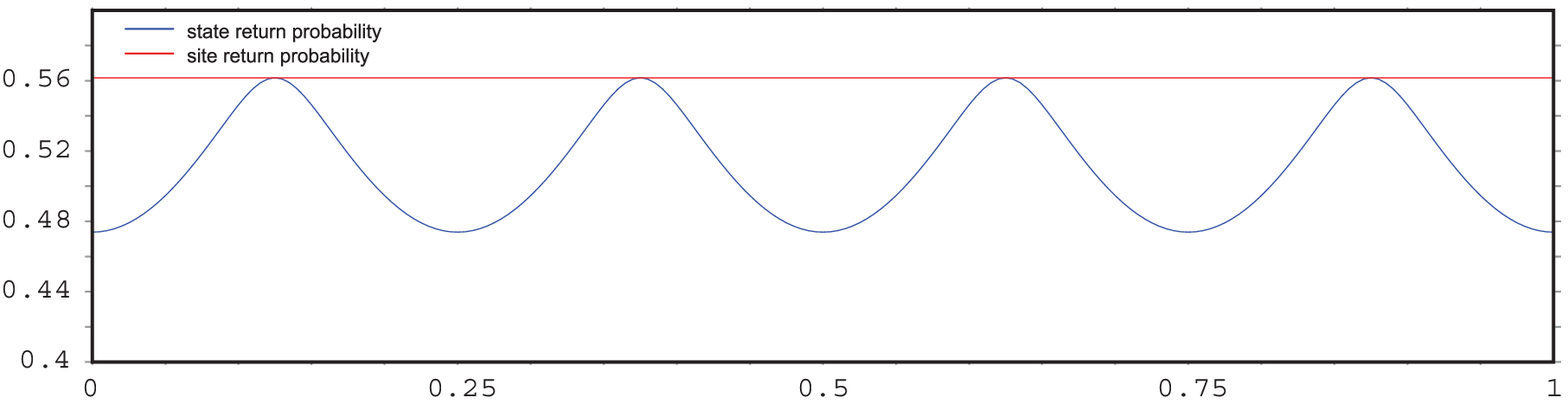}
 \caption{{\bf State versus site recurrence on the half-line and the line.}
 Given the qubits $\psi_x(t)=\cos t|x,\ua\>+\sin t|x,\da\>$, the upper figure represents as a function of $t$
 the state (blue) and site (red) return probabilities for a constant coin $\gamma_{2x}=\sqrt{3/5}$ on the half-line.
 The lower figure is a similar plot for the same coin on the line,
 but for the qubits $\psi_x(t)=\cos t\,|x,\ua\>+i\sin t\,|x,\da\>$.
  \label{fig:HL-L}}
\end{figure}

The upper plot in Figure~\ref{fig:HL-L} shows the state and site return probabilities for $\alpha=\cos t$ and $\beta=\sin t$ in the case of a
constant coin on the half-line given by $\gamma_{2x}=\sqrt{3/5}$. We see that this kind of coined walk exhibits the same
paradoxical behavior already
found in the toy model of Example~\ref{ex:SHIFT+EIG-Schur} since the state return probability is higher than the site return probability for some
values of $t$. By contrast, this effect is absent in the lower plot of Figure~\ref{fig:HL-L}, which represents the analogous situation on the line for
$\alpha=\cos t$ and $\beta=i\sin t$.

Indeed, for any constant coin $\gamma_{2x}=\gamma$ on the line, the above procedure for the computation of the state return probability for a qubit
$\psi_x=\alpha|x,\ua\>+\beta|x,\da\>$ yields $R_{\alpha,\beta}=\|f(z,0)f(z,c)\|^2$, where $f(z,\gamma_0)$ is the Schur function with Schur parameters
$(\gamma_0,0,\gamma,0,\gamma,0,\dots)$ and $c=\gamma-(2i\rho/\overline\gamma)\im(\overline\alpha\beta\gamma)$. On the other hand, our previous
results imply that the corresponding site return probability has the same value $\<\psi_x|\bs{R}_x\psi_x\>=\|f(z,0)\|^2$ for any qubit. Thus, for any
constant coin on the line, the state return probability of a qubit is never higher than its site return probability. Nevertheless, this can change if
we allow distinct coins on the line or if we consider absorbing subspaces not attached to a site.

\subsection{Site recurrence in two dimensions} \label{ssec:SITE-2D}

In this section we show some numerical results and display a few figures illustrating the behavior of
return probabilities starting from different states in a few types of two-dimensional coined walks.

We first consider the usual two-dimensional doubly-infinite square lattice $\Lambda=\mathbb{Z}^2$ with four internal degrees of freedom for each site
(east, north, west and south).  The Hilbert state space $\ell^2(\mathbb{Z}^2) \otimes \mathbb{C}^4$ is spanned by a basis of the form $|x,\ra\>$,
$|x,\ua\>$, $|x,\la\>$, $|x,\da\>$, $x\in\mathbb{Z}^2$. The unitary step $U=SC$ factorizes again into a shift $S|x,\ra\>$ $=|x+e_1,\ra\>$,
$S|x,\ua\>=|x+e_2,\ua\>$, $S|x,\la\>=|x-e_1,\la\>$, $S|x,\da\>=|x-e_2,\da\>$, with $e_1=(1,0)$ and $e_2=(0,1)$, and a coin operation $C$ that acts
independently on each site subspace $V_x=\spn\{|x,\ra\>,|x,\ua\>,|x,\la\>,|x,\da\>\}$ as a four-dimensional unitary operator.

The walks we consider are the well known Grover walk and the Fourier walk, defined respectively by the following constant coins in the site basis,
\[
 C_G^4 = \frac{1}{2}
       \begin{pmatrix}
       -1 & 1 & 1 & 1
       \\
       1 & -1 & 1 & 1
       \\
       1 & 1 & -1 & 1
       \\
       1 & 1 & 1 & -1
       \end{pmatrix},
 \qquad
 C_F^4 = \frac{1}{2}
       \begin{pmatrix}
       1 & 1 & 1 & 1
       \\
       1 & i & -1 & -i
       \\
       1 & -1 & 1 & -1
       \\
       1 & -i & -1 & i
       \end{pmatrix}.
\]

The method we use to produce these results is a classical one. It goes by the name of the ``transfer matrix method,'' and our best
references for
this material are the articles \cite{BP,BGPP} and the classical book \cite{STANLEY}. The main point is that once the four-dimensional unitary coin is
specified, one has a way of obtaining the four-dimensional return amplitudes $\bs\mu_n$ for any site.
The methods of this paper then allow one to compute many other quantities of interest.

\begin{figure}
 \includegraphics[width=12cm,height=5cm]{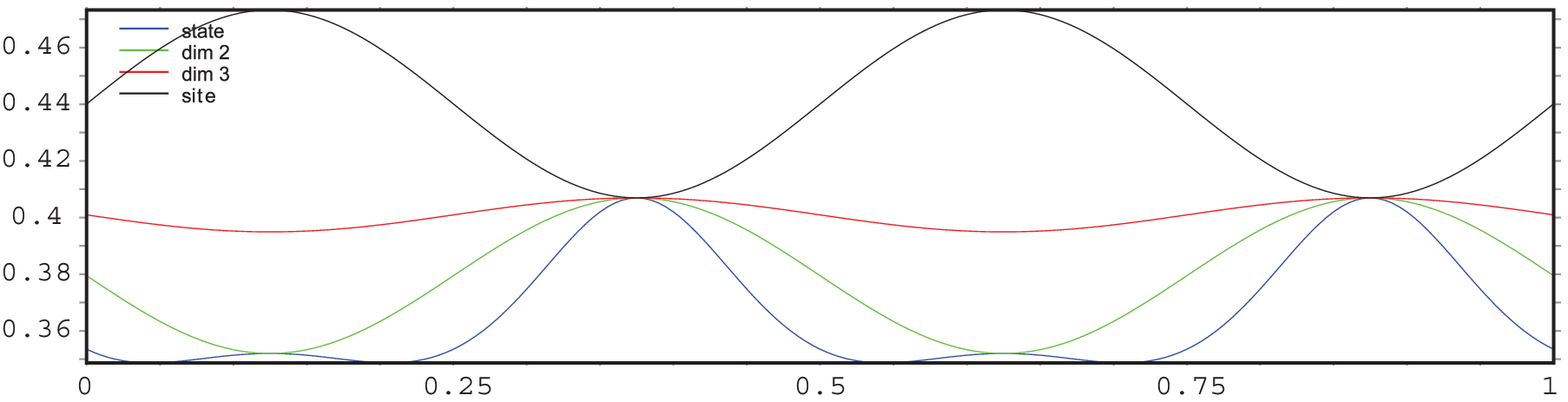}
 \includegraphics[width=12cm,height=5cm]{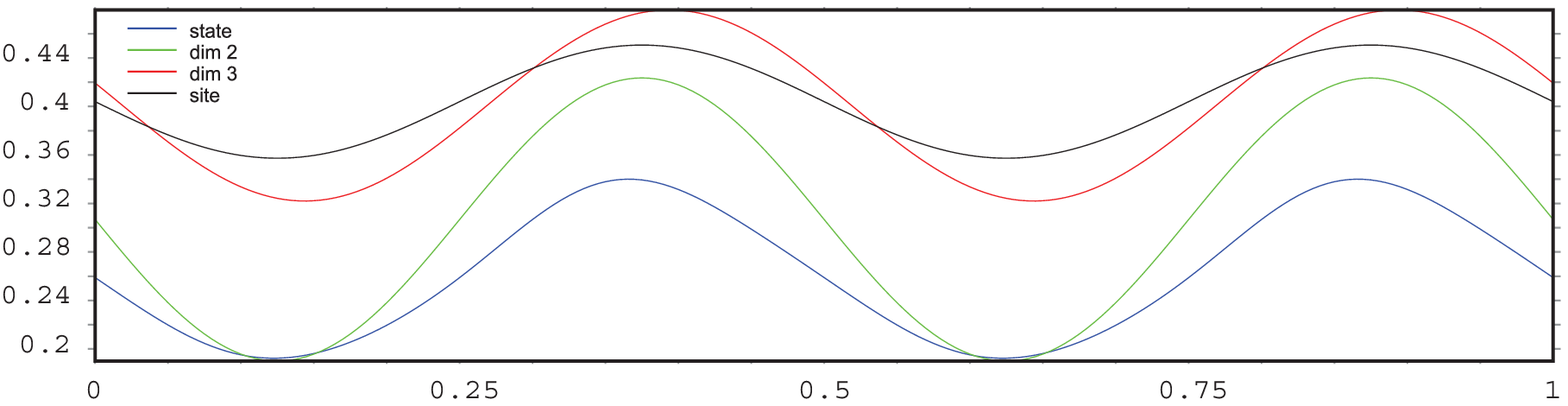}
 \includegraphics[width=12cm,height=5cm]{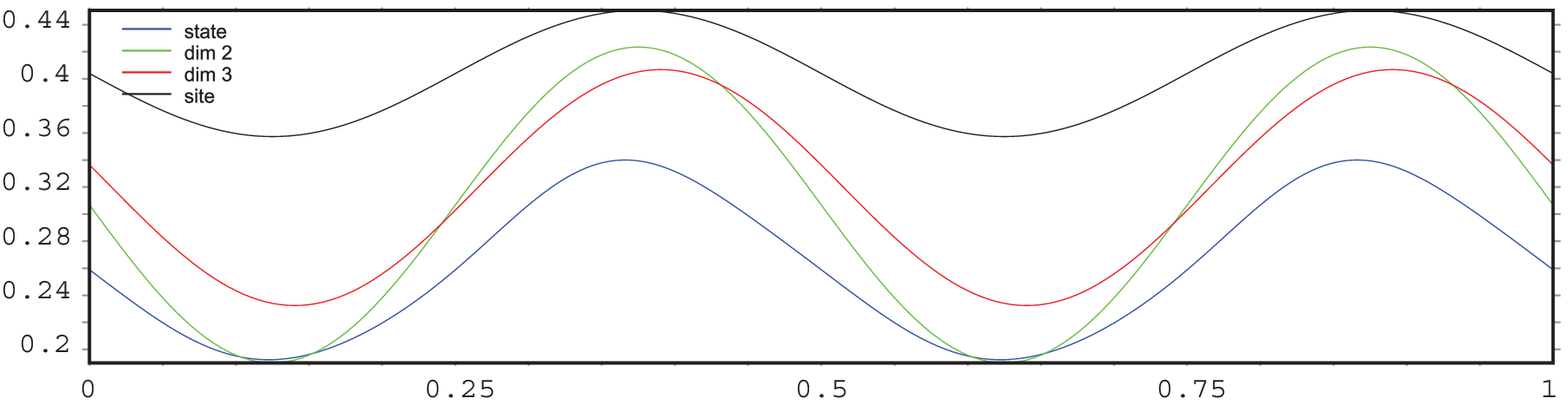}
 \caption{{\bf State versus site recurrence on a square lattice.}
 For a choice of qubits $\phi_x^1,\phi_x^2,\phi_x^3$ at site $x$,
 the figures show as a function of $t$ the state (blue) and site (black) return probability of $\psi_x(t)=\cos t\,\phi_x^1+\sin t\,\phi_x^2$,
 together with the return probability to the 2-dimensional subspace $\spn\{\phi_x^1,\phi_x^2\}$ (green)
 and the 3-dimensional subspace $\spn\{\phi_x^1,\phi_x^2,\phi_x^3\}$ (red).
 The upper figure represents the case of the Grover coin $C_G^4$ with
 $\phi_x^1=|x,\ra\>$, $\phi_x^2=|x,\la\>$, $\phi_x^3=\frac{1}{\sqrt{2}}(|x,\ua\>+i|x,\da\>)$.
 Curiously, the change $\phi_x^3=\frac{1}{\sqrt{2}}(|x,\ua\>+|x,\da\>)$ yields a 3-dimensional subspace
 with the same return probability plot as the whole site.
 The other two figures correspond to the Fourier coin $C_F^4$
 with the choice $\phi_x^1=|x,\ra\>$, $\phi_x^2=|x,\ua\>$, $\phi_x^3=\frac{1}{\sqrt{2}}(|x,\la\>+|x,\da\>)$ (center)
 and $\phi_x^3=\frac{1}{\sqrt{2}}(|x,\la\>+i|x,\da\>)$ (bottom).
}
 \label{fig:SQUARE}
\end{figure}

Figure~\ref{fig:SQUARE} compares the state and subspace return probabilities computed on a curve of qubits lying in the same site. Each figure
represents the return probabilities to several nested absorbing subspaces of the site subspace to reveal the possible exotic behavior of the return
probability. Although the Grover example shows a classical monotone behavior, it conceals an interesting effect because the return probability to the
four-dimensional site subspace coincides with the return probability to one of its three-dimensional subspaces. Furthermore, the two examples of the
Fourier walk show that the return probability is not a monotone function with respect to the inclusion relation among absorbing subspaces.

To see the effect of the geometry of the lattice in the return probabilities we also present in Figure~\ref{fig:HEXA} a similar comparison for a
couple of coined walks in a two-dimensional hexagonal lattice already used by other authors for different purposes \cite{KOLLAR}. In this
case each site has three internal degrees of freedom that control a conditional shift $S$ to the three nearest neighbors
analogously to the case of the square lattice. A three-dimensional unitary operator at each site defines a coin operation $C$ giving the unitary step
$U=SC$. The examples represented in Figure~\ref{fig:HEXA} correspond to the constant coins given in the site basis by the three-dimensional versions
of Grover and Fourier,
\[
 C_G^3 = \frac{1}{3}
       \begin{pmatrix}
       -1 & 2 & 2
       \\
       2 & -1 & 2
       \\
       2 & 2 & -1
       \end{pmatrix},
 \qquad
 C_F^3 = \frac{1}{\sqrt{3}}
       \begin{pmatrix}
       1 & 1 & 1
       \\
       1 & e^{i2\pi/3} & e^{-i2\pi/3}
       \\
       1 & e^{-i2\pi/3} & e^{i2\pi/3}
       \end{pmatrix}.
\]
These figures show that the exotic dependence of the return probability on the absorbing subspace holds for two-dimensional coined walks with a
constant coin for many lattice geometries. Indeed, the plots of Figure~\ref{fig:HEXA} show more dramatic phenomena than the previous ones.

\begin{figure}
 \includegraphics[width=12cm,height=5cm]{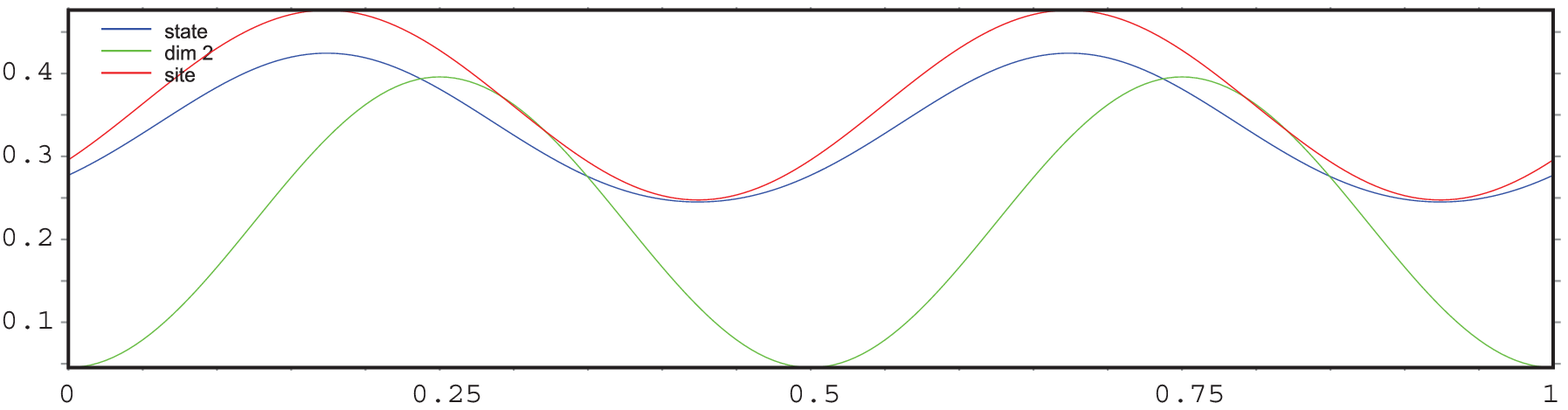}
 \includegraphics[width=12cm,height=5cm]{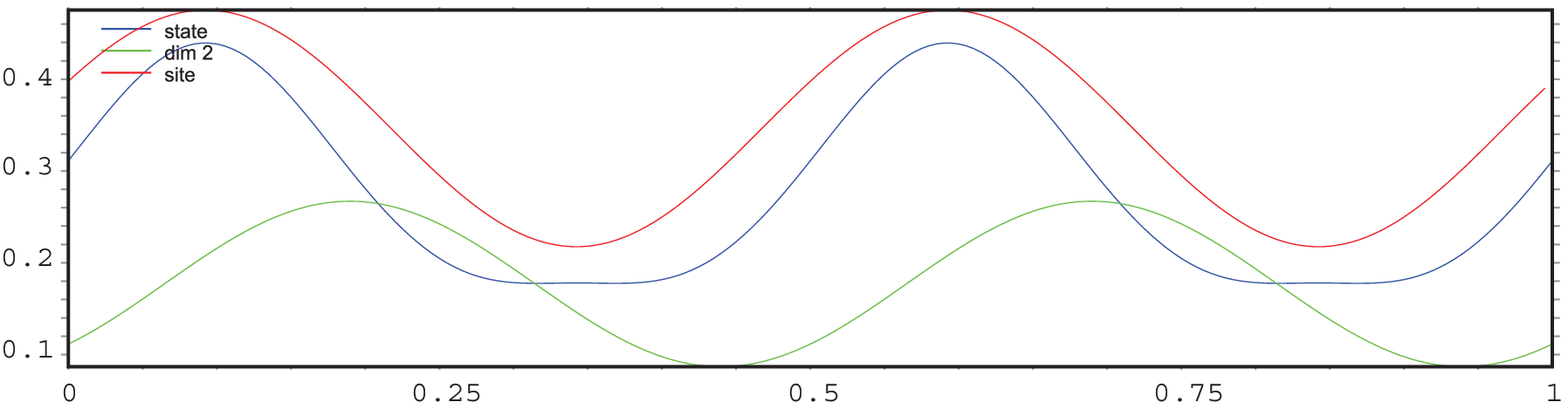}
 \caption{{\bf State versus site recurrence on a hexagonal lattice.}
 Different return probabilities as a function of $t$ for the qubits given in the site basis by
 $\psi_x(t)=\frac{1}{\sqrt{2}}\cos t\,(1,0,\frac{1+i}{\sqrt{2}})+i\sin t\,(0,1,0)$.
 Together with the state (blue) and site (red) return probability,
 we present the case of the 2-dimensional absorbing subspace
 $\spn\{\frac{1}{\sqrt{2}}(1,0,\frac{1+i}{\sqrt{2}}),(0,1,0)\}$ (green).
 The upper and lower figure correspond to the Grover coin, $C_G^3$, and the Fourier coin, $C_F^3$,
 respectively.
}
 \label{fig:HEXA}
\end{figure}

Each of these two-dimensional coined walks has a constant coin and no boundary conditions; thus,
the site return probability operator $\bs{R}=\bs{R}_x$
is independent of the site $x$. The computation of its eigenvalues reveals a curious doubling effect: there is always
an eigenvalue with multiplicity two.

\medskip

\begin{center}
\begin{tabular}{|l|c|l|}
 \hline
 \bf Lattice & \bf Coin & \bf Approximate eigenvalues of $\bs{R}$
 \\
 \hline
 Square & $C_G^4$ & 0.6593, 0.4069, 0.4069, 0.2878
 \\
 \hline
 Square & $C_F^4$ & 0.5517, 0.3882, 0.3882, 0.2880
 \\
 \hline
 Hexagonal & $C_G^3$ & 0.8017, 0.2411, 0.2411
 \\
 \hline
 Hexagonal & $C_0^3$ & 0.6365, 0.6365, 0.5462
 \\
 \hline
\end{tabular}
\end{center}

\medskip

\noindent Here $C_0^3$ is the three-dimensional coin
\[
 C_0^3 = \frac{1}{\sqrt{2}}
       \begin{pmatrix}
       0 & 0 & \sqrt{2}
       \\
       1 & 1 & 0
       \\
       1 & -1 & 0
       \end{pmatrix}.
\]

\noindent This points to a general property whose explanation and implications deserve future research.

\appendix

\section{The average of the expected $V$-return time} \label{app:TAU-AVERAGE}

\begin{lem} \label{lem:SURVIVAL-1}
If $\dim\mathcal{H}<\infty$ and $U$ has no eigenvectors in $V^\bot$, the eigenvalues of $\tilde{U}=(I-P)U$ and $A=(I-P)U(I-P)$ lie
inside the unit disk. Hence, $\|\tilde{U}^n\|_F$ and $\|A^n\|_F$ converge to zero exponentially fast as $n\to\infty$.
\end{lem}

\begin{proof}
Let $\psi\in\mathcal{H}\setminus\{0\}$ such that $\tilde{U}\psi=\lambda\psi$, $\lambda\in\mathbb{C}$. If $U\psi\in V^\bot$, then
$U\psi=\tilde{U}\psi=\lambda\psi$. Hence, if no eigenvector of $U$ belongs to $V^\bot$, we have $U\psi\notin V^\bot$.  But then
$|\lambda|\|\psi\|=\|\tilde{U}\psi\|<\|U\psi\|=\|\psi\|$, so $|\lambda|<1$.

Consider now $\phi\in\mathcal{H}\setminus\{0\}$ such that $A\phi=\lambda\phi$, $\lambda\in\mathbb{C}$. Then, $\psi=(I-P)\phi$ satisfies
$\tilde{U}\psi=(I-P)A\phi=\lambda\psi$. If $\psi=0$, then $A\phi=0$, i.e. $\lambda=0$. Otherwise, we are in the previous situation that leads to
$|\lambda|<1$.

The asymptotics of $\|\tilde{U}^n\|_F$ and $\|A^n\|_F$ follow directly from the previous results and the Jordan form of $\tilde{U}$ and $A$.
\end{proof}

A finite average of the expected $V$-return time requires a finitely supported spectral measure $\bs\mu(d\lambda)$ of $V$. We know that the mass
points of $\bs\mu(d\lambda)$ must be eigenvalues of $U$ and the mass at $\lambda$ is $\bs\mu(\{\lambda\})=PE(\{\lambda\})P$,
where $E(\{\lambda\})$ is the orthogonal projector onto the eigenspace associated with $\lambda$.

\begin{thm} \label{thm:TAU-MASS}
If $V$ has a finitely supported spectral measure $\bs\mu(d\lambda)$, then $\dashint_{S_V} \tau(\psi)\,d\psi = K/\dim V$ with $K$
a positive integer that can be computed equivalently as
\[
 K = \sum_k \dim(E_kV) = \sum_k \rank\bs\mu(\{\lambda_k\}),
\]
where $\{\lambda_k\}$ are the mass points of $\bs\mu(d\lambda)$ and $E_k=E(\{\lambda_k\})$ are
the orthogonal projectors onto the corresponding eigenspaces of $U$.
\end{thm}

\begin{proof}
Under the hypothesis, $V$ lies on a finite-dimensional $U$-invariant subspace $\mathcal{H}_0=\oplus_k E_kV$. Due to the unitarity of $U$, the
orthogonal subspace $\mathcal{H}_1=\mathcal{H}_0^\bot$ is also $U$-invariant and the unitary evolution splits as $U = U_0 \oplus U_1$ with $U_i = U
\upharpoonright \mathcal{H}_i$. Therefore, to compute the average of the expected $V$-return time, we can use the operator $U_0$ instead of $U$
since no state of $V$ escapes from $\mathcal{H}_0$ during the evolution. In other words, $\bs{a}_n^{U_0}=\bs{a}_n^U$ as operators on $V$.

So, in what follows we can assume without loss that $\mathcal{H}=\oplus_k E_kV$. Then, the eigenspaces of $\mathcal{H}$ are $E_kV$, which implies
that $U$ has no eigenvectors in $V^\bot$. Indeed, any eigenvector will have the form $\psi=E_k\phi$, $\phi\in V$, and the orthogonality requirement
$\psi\in V^\bot$ forces $\|\psi\|^2=\<\phi|E_k\phi\>=\<\phi|\psi\>=0$. This means that we can use the results of the previous lemma.

We know that $K=\sum_{n\ge1}n\|\bs{a}_n\|_F^2$. Since $\bs{a}_n=U\tilde{U}^{n-1}P-\tilde{U}^nP$, we have
$\|\bs{a}_n\|_F^2=\|\tilde{U}^{n-1}P\|_F^2-\|\tilde{U}^nP\|_F^2$. These identities combined with
$\lim_{n\to\infty}n\|\tilde{U}^n\|_F^2=0$ from the lemma, lead to $K = \sum_{n\ge0}\|\tilde{U}^nP\|_F^2$, which can also be expressed as
\[
 K = \|P\|_F^2 + \sum_{n\ge0}\|A^nB\|_F^2, \kern7pt A=(I-P)U(I-P), \kern7pt B=(I-P)UP.
\]
Using the equality $AA^\dag+BB^\dag=I-P$, we can write
\[
 \|A^nB\|_F^2 = \tr(A^n B B^\dag (A^\dag)^n) = \|A^n\|_F^2-\|A^{n+1}\|_F^2, \quad A^0=I-P.
\]
Thus, $\sum_{n\ge0}\|A^nB\|_F^2$ becomes a telescoping series whose convergence follows from the
lemma, giving finally
\[
 K = \|P\|_F^2 + \|I-P\|_F^2 = \dim\mathcal{H} = \sum_k\dim(E_kV).
\]
To finish the proof simply note that $\ker P \cap E_kV = V^\bot \cap E_kV = \{0\}$. This proves that
$\dim(E_kV)=\dim(PE_kV)=\rank(PE_kP)=\rank\bs\mu(\{\lambda_k\})$.
\end{proof}

The subspace $\mathcal{H}_0=\oplus_kE_kV$ appearing in the proof of this theorem is the minimal $U$-invariant subspace
containing $V$. The numerator $K$ is its dimension.

In the case of a finite-dimensional state space $\mathcal{H}$ we can also express $K=\dim\mathcal{H}-\nu$, where $\nu$ is the number of linearly
independent eigenvectors of $U$ in $V^\bot$. To see this, note that we can consider the points $\lambda_k$ running over all the eigenvalues of $U$ --
i.e., over all the mass points of $E(d\lambda)$ -- bearing in mind that $\bs\mu(\{\lambda_k\})=PE_kP$ will vanish when $\lambda_k$ is not a mass
point of $\bs\mu(d\lambda)$. Then, denoting by $V_k=E_k\mathcal{H}$ the eigenspace of $U$ associated with $\lambda_k$, it is clear that
\[
 \nu = \sum_k \dim(V_k \cap V^\bot).
\]
A simple calculation shows that $(E_kV)^{\bot V_k} = V_k \cap V^\bot$, which means that $V_k = E_kV \oplus (V_k \cap V^\bot)$ is an orthogonal
decomposition of the eigenspace $V_k$. Hence, $\dim V_k = \dim E_kV + \dim(V_k \cap V^\bot)$. If $\dim \mathcal{H} < \infty$ the eigenvectors of the
unitary operator $U$ span $\mathcal{H}$; thus, $\dim \mathcal{H} = \sum_k \dim V_k$ and
\[
 K = \sum_k \dim(E_kV) = \dim \mathcal{H} - \nu.
\]

The reason we can use Lemma \ref{lem:SURVIVAL-1} in the proof of the previous theorem -- even if $\dim\mathcal{H}=\infty$ -- is that
the lemma can be applied to the minimal $U$-invariant subspace $\mathcal{H}_0$ containing $V$ whenever $\dim\mathcal{H}_0<\infty$,
i.e.~if all the states of $V$ have a finite expected $V$-return time.
In that case, the unitary operator $U_0 = U
\upharpoonright \mathcal{H}_0$ and the orthogonal projector $P_0$ onto $\mathcal{H}_0$ satisfy
\[
  \tilde{U}^nP_0=((P_0-P)U_0)^n=\tilde{U}_0^n \quad \text{and} \quad
  A^nP_0=((P_0-P)U_0(P_0-P))^n=A_0^n.
\]
Thus, we find an exponentially fast convergence $\|\tilde{U}^nP_0\|_F,\|A^nP_0\|_F\to0$, which implies that
$\|\tilde{U}^nP\|_F\to0$ exponentially fast as well. Nevertheless, Lemma \ref{lem:SURVIVAL-1} does not ensure these asymptotics for
states with a finite expected $V$-return time that do not cover all of $S_V$. Even in this case, the following
result shows that a weaker asymptotic result holds.

\begin{lem} \label{lem:SURVIVAL-2}
Any $V$-recurrent state $\psi$ with $\tau(\psi)<\infty$ satisfies
\[
 \lim_{n\to\infty}n\|\tilde{U}^n\psi\|^2=0.
\]
\end{lem}

\begin{proof}
For any $\psi \in S_V$, the identity $\|\bs{a}_n\psi\|^2=\|\tilde{U}^{n-1}\psi\|^2-\|\tilde{U}^n\psi\|^2$ yields
\[
 \tau=\tau(\psi)=\sum_{n\geq1}n(s_{n-1}-s_n)=\lim_{n\to\infty}\left(\sum_{k=0}^{n-1}s_k-ns_n\right), \quad s_n=\|\tilde{U}^n\psi\|^2.
\]

On the other hand, the $V$-recurrence of $\psi$ can be characterized in terms of the non-increasing sequence $s_n$ of $V$-survival probabilities in
any of the following equivalent ways
\[
 s_n\to0
 \quad \Leftrightarrow \quad
 \prod_{n\geq1}\frac{s_n}{s_{n-1}}=0
 \quad \Leftrightarrow \quad
 \sum_{n\geq1}\left(1-\frac{s_n}{s_{n-1}}\right)=\infty.
\]

Therefore, if $\psi$ is $V$-recurrent and $\tau<\infty$ the series $\sum_{n\geq0}s_n$ must converge. Otherwise, $ns_n\to\infty$ and thus, for some
index $n_0$,
\[
 \tau = \sum_{n\geq1} ns_{n-1} \left(1-\frac{s_n}{s_{n-1}}\right) \geq
 \sum_{n\geq n_0} \left(1-\frac{s_n}{s_{n-1}}\right) = \infty.
\]
Hence, $\sum_{n\geq0}s_n=\tau'<\infty$ and $ns_n \to \tau'-\tau$. Finally, $\tau'-\tau\neq0$ leads to $s_n \sim (\tau'-\tau)/n$, in
contradiction with the convergence of $\sum_{n\geq0}s_n$.
\end{proof}

\section{Site Schur functions for one-dimensional coined walks} \label{app:SITE-SCHUR}

The purpose of this appendix is to calculate the Schur function $\bs{f}_{\kern-2pt x}(z)$ of a site subspace $V_x=\{|x,\ua\>,|x,\da\>\}$ for a
one-dimensional coined walk with arbitrary coins. We will perform this calculation first for a semi-infinite lattice, extending it afterwards to a
finite lattice and finally to a doubly-infinite lattice.

On the half-line, $\bs{f}_{\kern-2pt x}(z)$ is the Schur function of the matrix measure
\[
 \bs\mu_x(d\lambda) =
 \begin{pmatrix}
 |X_{2x}(\lambda)|^2 & \overline{X_{2x}(\lambda)} X_{2x+1}(\lambda)
 \\
 X_{2x}(\lambda) \overline{X_{2x+1}(\lambda)} & |X_{2x+1}(\lambda)|^2
 \end{pmatrix}
 \mu(d\lambda),
\]
where $X_k(z)$ are the orthogonal Laurent polynomials generated by the CMV matrix with Verblunsky parameters $(\gamma_0,0,\gamma_2,0,\dots)$ and
$\mu(d\lambda)$ is the corresponding orthogonality measure.

The above matrix measure can be rewritten in terms of the orthogonal polynomials $\varphi_k(z)$ with respect to $\mu(d\lambda)$ and their reversed
ones $\varphi_k^*(z)=z^k\overline{\varphi_k}(z^{-1})$ using the relations \cite{FIVE,SIMON1,SIMONfoot,SIMON5years,WATKINS}
\[
 X_{2x}(z)=z^{-x}\varphi_{2x}(z), \qquad X_{2x+1}(z)=z^{-x-1}\varphi_{2x+1}^*(z),
\]
which yield
\[
 \bs\mu_x(d\lambda) =
 \begin{pmatrix}
 |\varphi_{2x}(\lambda)|^2 & \lambda^{-2x-1}\varphi_{2x}^*(\lambda)\,\varphi_{2x+1}^*(\lambda)
 \\
 \lambda^{-2x}\varphi_{2x}(\lambda)\,\varphi_{2x+1}(\lambda) & |\varphi_{2x+1}(\lambda)|^2
 \end{pmatrix}
 \mu(d\lambda),
\]

The first step is to compute the Carath\'{e}odory function $\bs{F}_{\kern-2pt x}(z)$ of $\bs\mu_x(d\lambda)$,
\[
 \bs{F}_{\kern-2pt x}(z) = \begin{pmatrix} F_{2x}(z) & \tilde{G}_{2x}(z) \\ G_{2x}(z) & F_{2x+1}(z) \end{pmatrix},
\]
where $F_k(z)$ is the Carath\'{e}odory function of $|\varphi_k(z)|^2\mu(d\lambda)$ and
\[
\begin{aligned}
 & G_k(z) = \int \frac{\lambda+z}{\lambda-z} \, \lambda^{-k}\varphi_k(\lambda)\,\varphi_{k+1}(\lambda) \, \mu(d\lambda),
 \\
 & \tilde{G}_k(z) = \int \frac{\lambda+z}{\lambda-z} \, \lambda^{-k-1}\varphi_k^*(\lambda)\,\varphi_{k+1}^*(\lambda) \, \mu(d\lambda).
\end{aligned}
\]

Let $f(z)$ be a Schur function with Schur parameters $(\gamma_k)$. Khrushchev's formula states that \cite{KHRUSHCHEV}
\[
 F_k(z) = \frac{1+zf^{k-1}(z)f_k(z)}{1-zf^{k-1}(z)f_k(z)}.
\]
We will use the notation $F(z)=F_0(z)$, which is the Carath\'{e}odory function related to $f(z)$.

On the other hand, $G_k(z)$ and $\tilde{G}_k(z)$ can be expressed in terms of the second kind polynomials $\Omega_k(z)$, generated by the Verblunsky
parameters $(-\gamma_k)$, due to the relations \cite{JNT}
\[
\begin{aligned}
 & \int \frac{\lambda+z}{\lambda-z} \, \left[ \frac{z^j}{\lambda^j}\varphi_k(\lambda) - \varphi_k(z) \right]\mu(d\lambda) = \Omega_k(z),
 && \quad j=0,1,\dots,k-1,
 \\
 & \int \frac{\lambda+z}{\lambda-z} \, \left[ \frac{z^j}{\lambda^j}\varphi_k^*(\lambda) - \varphi_k^*(z) \right]\mu(d\lambda) = -\Omega_k^*(z),
 && \quad j=1,2,\dots,k.
\end{aligned}
\]
These relations lead to the equalities
\[
\begin{aligned}
 & G_k(z) = z^{-k}\varphi_k(z) \left[\varphi_{k+1}(z)F(z)+\Omega_{k+1}(z)\right],
 \\
 & \tilde{G}_k(z) = z^{-k-1}\varphi_k^*(z) \left[\varphi_{k+1}^*(z)F(z)-\Omega_{k+1}^*(z)\right].
\end{aligned}
\]

The Verblunsky parameters provide a forward and a backward recurrence relation involving the orthogonal polynomials $\varphi_k(z)$ and their reversed
ones \cite{JNT,SIMON1,SIMONfoot},
\[
\begin{aligned}
 & \rho_k\varphi_{k+1}(z) = z\varphi_k(z) - \overline\gamma_k\varphi_k^*(z),
 && \quad \rho_kz\varphi_k(z) = \varphi_{k+1}(z) + \overline\gamma_k\varphi_{k+1}^*(z),
 \\
 & \rho_k\varphi_{k+1}^*(z) = -\gamma_kz\varphi_k(z) + \varphi_k^*(z),
 && \quad \rho_k\varphi_k^*(z) = \gamma_k\varphi_{k+1}(z) + \varphi_{k+1}^*(z),
\end{aligned}
\]
with analogous formulas for $\Omega_k(z)$, replacing $\gamma_k$ by $-\gamma_k$. From these recurrence
relations we get by induction an expression for the iterates and inverse iterates of $f(z)$, namely
\[
 f_k(z) = \frac{1}{z} \frac{\varphi_k^*(z)F(z)-\Omega_k^*(z)}{\varphi_k(z)F(z)+\Omega_k(z)}, \qquad f^k(z)=\frac{\varphi_{k+1}(z)}{\varphi_{k+1}^*(z)}.
\]
Therefore, $G_k(z)$ and $\tilde{G}_k(z)$ are related by
\[
 \frac{\tilde{G}_k(z)}{G_k(z)} = \frac{f_{k+1}(z)}{f^{k-1}(z)}.
\]

A simple connection among the functions $F_k(z)$, $F_{k+1}$, $G_k(z)$ and $\tilde{G}_k(z)$ results from rewriting
\[
 F_k(z) = \int \frac{\lambda+z}{\lambda-z} \, \lambda^{-k}\varphi_k^*(z)\,\varphi_k(z)\,\mu(d\lambda).
\]
Using this expression together with the recurrence for the orthogonal polynomials yields
\[
 \rho_k (F_k(z) - F_{k+1}(z)) = \gamma_k G_k(z) + \overline\gamma_k \tilde{G}_k(z).
\]
Finally, the above identity combined with the relation between $G_k(z)$ and $\tilde{G}_k(z)$ leads to
\[
\begin{aligned}
 & G_k(z) = \rho_k f^{k-1}(z) \frac{F_k(z)-F_{k+1}(z)}{\gamma_kf^{k-1}(z)+\overline\gamma_kf_{k+1}(z)},
 \\
 & \tilde{G}_k(z) = \rho_k f^{k+1}(z) \frac{F_k(z)-F_{k+1}(z)}{\gamma_kf^{k-1}(z)+\overline\gamma_kf_{k+1}(z)}.
\end{aligned}
\]

Bearing in mind Khrushchev's formula, it is obvious that, as in the case of $F_k(z)$, we can express $G_k(z)$ and $\tilde{G}_k(z)$ in terms of
iterates and inverse iterates of $f(z)$. Indeed, such an expression can be simplified using the Schur algorithm to relate different iterates or
different inverse iterates, giving
\[
\begin{aligned}
 & G_k(z) = \frac{2\rho_k zf^{k-1}(z) }{(1-\gamma_kzf^{k-1}(z))(1-zf^k(z)f_{k+1}(z))},
 \\
 & \tilde{G}_k(z) = \frac{2\rho_k zf_{k+1}(z) }{(1-\gamma_kzf^{k-1}(z))(1-zf^k(z)f_{k+1}(z))}.
\end{aligned}
\]

We have rewritten the coefficients of the matrix Carath\'{e}odory function $\bs{F}_{\kern-2pt x}(z)$ in terms of iterates and inverse iterates of the
scalar Schur function $f(z)$. Now it only remains to perform a tedious calculation using again the Schur algorithm for these iterates to arrive at
\[
 \bs{f}_{\kern-2pt x}(z) = z^{-1} (\bs{F}_{\kern-2pt x}(z)-I_2)(\bs{F}_{\kern-2pt x}(z)+I_2)^{-1}
 = \begin{pmatrix} \gamma_{2x}f^{2x-1}(z) & \rho_{2x}f_{2x+1}(z) \\ \rho_{2x}f^{2x-1}(z) & -\overline\gamma_{2x}f_{2x+1}(z) \end{pmatrix}.
\]

In the case of a finite lattice of $N$ sites, the only difference from the semi-infinite lattice is that the measure $\mu(d\lambda)$ is supported on
$2N$ points and thus there exists only a finite segment of $2N$ orthogonal polynomials and Laurent polynomials. The only effect of this finiteness in
the previous arguments is that the Schur function $f(z)$ has a finite number of Schur parameters $(\gamma_0,\dots,\gamma_{2N-2},1)$, and similarly
for its iterates. Hence, the site Schur function has the same form as in the half-line case.

As for a doubly-infinite lattice, we can obtain the corresponding site Schur functions by a limiting process on semi-infinite lattices. Consider a
coined walk on the line with coins given by a doubly-infinite sequence of parameters $(\gamma_{2x})_{x\in\mathbb{Z}}$. The transition matrix in the
basis states is then a doubly-infinite CMV matrix with Verblunsky parameters $(\dots,0,\gamma_{-2},0,\gamma_0,0,\gamma_2,0,\dots)$.

If $U=SC$ is the corresponding unitary step, let us define for each site $x_0\in\mathbb{Z}$ a new unitary $U_{x_0} = I_{x_0}^- \oplus U_{x_0}^+$
where $I_{x_0}^-$ is the identity on the strict left subspace $\oplus_{x<x_0}V_x$, and $U_{x_0}^+=S_{x_0}^+C$ acts on the right subspace
$\oplus_{x\ge x_0}V_x$ with the same coin operator as $U$, but a minimally modified shift $S_{x_0}^+$ to impose a reflecting boundary condition from
the right at site $x_0$: $S_{x_0}^+|x_0,\da\>=|x_0,\ua\>$.
In other words, $U_{x_0}$ splits into the trivial evolution on the strict left subspace and a coined walk on a half-line -- the sites $x\ge x_0$ --
whose unitary step acts exactly as $U$ on the strict right subspace $\oplus_{x>x_0}V_x$, although not on the subspace $V_{x_0}$
itself. Clearly, $U_{x_0}$ converges strongly to $U$ as $x_0 \to -\infty$, i.e. $\lim_{x_0\to-\infty}(U-U_{x_0})\psi = 0$ for any state $\psi$.

Additionally, the analysis of the half-line shows that, under the $U_{x_0}$ evolution, the matrix Schur function of a site $x\ge x_0$ is given by
\[
 \bs{f}_{\kern-2pt x_0,x}(z) =
 \begin{pmatrix} \gamma_{2x}f^{2x-1}_{x_0}(z) & \rho_{2x}f_{2x+1}(z) \\ \rho_{2x}f^{2x-1}_{x_0}(z) & -\overline\gamma_{2x}f_{2x+1}(z) \end{pmatrix},
\]
with $f_k(z)$ and $f^k_{x_0}(z)$ given by the Schur parameters
\[
  (\gamma_k,\gamma_{k+1},\gamma_{k+2},\dots)
  \quad \text{and} \quad
  (-\overline\gamma_k,-\overline\gamma_{k-1},\dots,-\overline\gamma_{2x_0},1),
\]
respectively.

The strong convergence $U_{x_0} \to U$ implies the weak convergence of the spectral measures $\bs\mu_{x_0,x}(d\lambda)\to\bs\mu_x(d\lambda)$
corresponding to any site subspace $V_x$, which also implies the uniform convergence $\bs{f}_{\kern-2pt x_0,x}(z)\to\bs{f}_{\kern-2pt x}(z)$ on
compact subsets of the unit disk for the related Schur functions. Hence,
\[
 \bs{f}_{\kern-2pt x}(z) =
 \begin{pmatrix} \gamma_{2x}f^{2x-1}(z) & \rho_{2x}f_{2x+1}(z) \\ \rho_{2x}f^{2x-1}(z) & -\overline\gamma_{2x}f_{2x+1}(z) \end{pmatrix},
\]
with $f^{2x-1}(z) = \lim_{x_0\to-\infty} f^{2x-1}_{x_0}(z)$. As a limit of Schur functions under the topology of uniform convergence on compact sets,
$f^{2x-1}(z)$ is a Schur function too and its Schur parameters are the limits of the Schur parameters of $f^{2x-1}_{x_0}(z)$. Taking the limit
$x_0\to-\infty$ in the Schur parameters $(0,-\overline\gamma_{2x-2},0,-\overline\gamma_{2x-4},\dots,-\overline\gamma_{2x_0},1)$ of
$f^{2x-1}_{x_0}(z)$ shows that $f^{2x-1}(z)$ is characterized by the infinite sequence of Schur parameters
$(0,-\overline\gamma_{2x-2},0,-\overline\gamma_{2x-4},\dots)$.

\vskip20pt

\noindent{\bf Acknowledgements}

\medskip

J. Bourgain acknowledges support from the National Science Foundation through grants DMS-0808042 and DMS-0835373, and thanks the Mathematics
Department at UC Berkeley for their hospitality.

F.A. Gr\"{u}nbaum acknowledges support from the Applied Math. Sciences subprogram of the Office of Energy Research, US Department of Energy, under
Contract DE-AC03-76SF00098.

The work of L. Vel\'{a}zquez was partly supported by the research project MTM2011-28952-C02-01 from the Ministry of Science and Innovation of Spain and
the European Regional Development Fund (ERDF), and by Project E-64 of Diputaci\'on General de Arag\'on (Spain).

J. Wilkening was supported in part by the Director, Office of Science, Computational and Technology Research, U.S. Department of Energy, under
Contract No.~DE-AC02-05CH11231, and by the National Science Foundation through grant DMS-0955078.

\end{document}